\documentclass[aps,pre,twocolumn,superscriptaddress,10pt]{revtex4-2}
\usepackage[utf8]{inputenc}
\usepackage{multirow,amsthm,amssymb,amsbsy,amsmath,epsfig,epstopdf,bm,float}
\usepackage{bm}
\usepackage{graphicx}
\usepackage{booktabs}
\usepackage{verbatim}
\usepackage{dcolumn}
\usepackage{color}
\usepackage{xcolor}
\usepackage[normalem]{ulem}
\usepackage{soul}
\usepackage{braket}
\usepackage{subfig}
\usepackage[export]{adjustbox}
\usepackage{enumerate}

\newcommand{\bb}[1]{\mathbb{#1}}
\newcommand{\tr}{\text{tr}}
\newcommand{\vertiii}[1]{{\left\vert\kern-0.25ex\left\vert\kern-0.25ex\left\vert #1 
\right\vert\kern-0.25ex\right\vert\kern-0.25ex\right\vert}}

\newtheorem{theorem}{Theorem} 
\newtheorem{definition}[theorem]{Definition}
\newtheorem{lemma}[theorem]{Lemma}
\newtheorem{remark}[theorem]{Remark}
\newtheorem{proposition}[theorem]{Proposition}
\newtheorem{corollary}[theorem]{Corollary}
\newtheorem{example}[theorem]{Example}

\begin{document}
\title{Input-dependence in quantum reservoir computing}
\author{Rodrigo Mart\'inez-Pe\~na}
\email{rodrigo.martinez@dipc.org}
\affiliation{Donostia International Physics Center, Paseo Manuel de Lardizabal 4, E-20018 San Sebastián, Spain}
\author{Juan-Pablo Ortega}
\email{juan-pablo.ortega@ntu.edu.sg}
\affiliation{School of Physical and Mathematical Sciences, Nanyang Technological University, 21 Nanyang Link, Singapore 637371}

\date{ \today }

\begin{abstract}
Quantum reservoir computing is an emergent field in which quantum dynamical systems are exploited for temporal information processing. Previous work found a feature that makes a quantum reservoir valuable: contractive dynamics of the quantum reservoir channel toward input-dependent fixed points. These results are enhanced in this paper by finding conditions that guarantee a crucial aspect of the reservoir's design: distinguishing between different input sequences to ensure a faithful representation of temporal input data. This is implemented by finding a condition that guarantees injectivity in reservoir computing filters, with a special emphasis on the quantum case. We provide several examples and focus on a family of quantum reservoirs that is much used in the literature; it consists of an input-encoding quantum channel followed by a strictly contractive channel that enforces the echo state and the fading memory properties. This work contributes to analyzing valuable quantum reservoirs in terms of their input dependence. 
\end{abstract}	


\maketitle

\section{Introduction}

Quantum reservoir computing (QRC) is an active research area in quantum machine learning. When applied to time series processing, QRC uses complex quantum systems to process the input information without tuning its physical parameters or, at most, a reduced number of them. The connection with traditional reservoir computing (RC) is that the reservoir is left untouched during training, which involves exclusively running a simple (regularized) linear regression on a set of output observables. Researchers are exploring a wide variety of quantum dynamical systems to perform this technique, such as photonic platforms \cite{garcia2024squeezing,garcia2025quantum,nokkala2024retrieving,de2024recurrent}, unitaries with dynamical phase transitions \cite{martinez2021dynamical,palacios2024role,xia2022reservoir,ivaki2024quantum,llodra2025quantum} and driven open quantum systems \cite{dudas2023quantum,sannia2024dissipation,sannia2024skin,abbas2024reservoir,zhu2024practical}. This rapid increase in relevance has been compelled by the recent experimental implementation of QRC models in noisy quantum devices  \cite{dasgupta2022characterizing,mlika2023user,suzuki2022natural,kubota2023temporal,yasuda2023quantum,chen2020temporal,kubota2023temporal,molteni2023optimization,pfeffer2022hybrid,ahmed2025optimal,hu2024overcoming,de2024recurrent,miao2024quantumreservoirpy}.

The growth of the QRC field has been accompanied by an increase in the number of theoretical research works: efficient prediction of quantum many-body dynamics \cite{sornsaeng2024quantum}, experimental protocol designs \cite{pfeffer2022hybrid,mujal2023time,franceschetto2024harnessing,kobayashi2024feedback,hu2024overcoming},  
universal approximation properties
\cite{chen2019learning,chen2020temporal,nokkala2021gaussian,sannia2024dissipation,sannia2024skin}, and the design of useful and operational QRC models \cite{martinez2023quantum,kobayashi2024extending,kobayashi2024coherence}. In this paper, we focus on the last question, that is, the design of operational QRC models. In particular, we concentrate on a property of reservoir computing models that has not been much considered so far in the QRC literature: the injectivity of reservoir filters. An injective reservoir filter is one that allows us to distinguish inputs not by looking at the outputs at each time step but by considering the whole input and output sequences. The ability to distinguish input sequences in this sense is a necessary condition for good performance in many classification and regression temporal information processing tasks. An important particular case is the learning of deterministic dynamical systems using reservoirs. As we shall recall later on in more detail, it has been recently shown \cite{grigoryeva2021chaos, RC21, RC26} that a sufficient condition for the learnability of a dynamical system is the existence of an injective generalized synchronization \cite{kocarev1996generalized} between the dynamical system observations available as data and the reservoir that is used to learn that system. We see later on how the injectivity of the generalized synchronization necessarily implies that of the reservoir filter used in the learning task. This makes of filter injectivity an important architecture constraint that has to be kept in mind at the time of QRC systems design.

At this point, it is important to emphasize that filter injectivity is not related to either the (short term) memory \cite{Jaeger:2002, RC15, RC23} or the expressivity \cite{grigoryeva2018universal, grigoryeva2018echo, gonon2019reservoir, RC12} of the model, two different questions that have also been extensively explored in depth in the literature in connection with reservoir design. Filter injectivity only encompasses the set-theoretical one-to-one correspondence between input and output sequences. Note also that this condition is stronger than the usual separability in the Stone-Weierstrass Theorem traditionally used to conclude the universality of various RC families \cite{DynamicalSystemsMaass, maass2, grigoryeva2018universal} as that condition requires to separate each pair of inputs with an arbitrary element of a (potentially large) RC family. In our case, injectivity is a separability property of a given and fixed filter. 

Our analysis will be conducted in the QRC context and will exploit the state-affine system (SAS) representation of QRC systems introduced in Ref.~\cite{martinez2023quantum} to study the injectivity question. Several examples will be provided to complement the discussion.
Particular attention will be paid to a broad family of quantum models obtained as the composition of two general quantum channels: first, a map that feeds the input to the quantum system, and second, another one that ensures some of the basic requirements for performing reservoir computing, that is, the echo state property (ESP) and the fading memory property (FMP). We call these systems {\it contracted-encoded quantum channels}.

The structure of the paper is as follows. Section \ref{sec:definitions} introduces the notation, general framework, and definitions of quantum reservoir systems. At the end of this section, in Theorem \ref{th:1} we generalize Theorem 2 in Ref.~\cite{martinez2023quantum} to any state-space map that is continuous and a contraction in the state-space entry. Sec. \ref{sec:results} contains all the results about injectivity in the paper, including proofs and examples: Lemma \ref{lm:1}, together with \eqref{maximal rank cond}, give sufficient conditions for global filter injectivity; Proposition \ref{local proposition 1} focuses on local filter injectivity around constant output sequences, with Proposition \ref{characterization constant outputs} characterizing the input sequences that lead to these outputs; Propositions \ref{prop:1} and \ref{prop:2}, and Corollary \ref{cor:1} particularize this analysis to the SAS model; Sec. \ref{sec:c-e quatum channels} concentrates on the contracted-encoded quantum channels, with Theorem \ref{th:2} describing situations that yield constant filters, and later extending the injectivity results to this family. Finally, Sec. \ref{sec:discussion} summarizes the main conclusions that can be drawn from the paper.

\section{Notation, definitions, and preliminary discussion}\label{sec:definitions}

We will work under the same assumptions as in Ref.~\cite{martinez2023quantum}, introducing only the main definitions for the sake of completeness. Consider a complex finite-dimensional Hilbert space $\mathcal{H}$. The space of quantum {\it density matrices} $\mathcal{S}(\mathcal{H})$ is a compact convex subset defined by 
\begin{equation}
\mathcal{S}(\mathcal{H})=\{\rho\in\mathcal{B}(\mathcal{H})\ |\ \rho^{\dagger}=\rho,\ \rho\geq 0,\ \tr(\rho)=1\},
\end{equation} 
where $\mathcal{B}(\mathcal{H})$ is the set of all bounded operators and $\rho^\dagger$ is the conjugate transpose. A quantum channel is a linear map $T: \mathcal{B}(\mathcal{H})\rightarrow \mathcal{B}(\mathcal{H})$ that is {\it completely positive and trace-preserving} (CPTP).  Completely positive and trace-preserving maps leave $\mathcal{S}(\mathcal{H})$ invariant and hence induce a restricted map $T: \mathcal{S}(\mathcal{H})\longrightarrow\mathcal{S}(\mathcal{H})$ that we shall denote with the same symbol and use interchangeably. 

It can be shown that any CPTP map $T: \mathcal{S}(\mathcal{H})\longrightarrow\mathcal{S}(\mathcal{H})$ is {\it nonexpansive} in the trace norm, which means that after applying $T$ to two input states $\rho_1,\rho_2 \in \mathcal{S}(\mathcal{H})$, the distance between these density matrices is either contracted or remains equal:
\begin{equation*}
||T(\rho_1)-T(\rho_2)||_{1}\leq||\rho_1-\rho_2||_1, 
\end{equation*}
where $||A||_1:=\tr\sqrt{AA^{\dagger}}$. We say that a quantum channel is {\it strictly contractive} when
\begin{equation}
\label{contraction}
||T(\rho_1)-T(\rho_2)||_{1}\leq r||\rho_1-\rho_2||_1
\end{equation}
for all $\rho_1,\rho_2\in \mathcal{S}(\mathcal{H})$, where $0\leq r<1$. Strictly contractive maps can be constructed by composing a strictly contractive channel with a general CPTP map, a fact that will be exploited later on and that will motivate the study of a particular family of QRC systems.

We now introduce the quantum reservoir computing framework. A QRC system is specified by a state equation and a readout map. The state equation is a family of continuous CPTP maps $T:\mathcal{B}(\mathcal{H})\times \bb{R}^n\rightarrow \mathcal{B}(\mathcal{H})$, with $n\in \bb{N}$ the number of input features, which are taken to be real values (classical inputs).  The readout map $h:\mathcal{B}(\mathcal{H})\rightarrow \bb{R}^m$, with $m\in \bb{N}$, maps operators in $\mathcal{B}(\mathcal{H})$ to the Euclidean space $\bb{R}^m$ and is used to produce the output of the QRC system. 
Inputs are bi-infinite discrete-time sequences of the form $\underline{\textbf{z}}=(\dots,{\bf z}_{-1},{\bf z}_0,{\bf z}_1,\dots)\in (\bb{R}^n)^{\bb{Z}}$, and outputs $\underline{\textbf{y}}\in(\bb{R}^m)^{\bb{Z}}$ have the same structure. A QRC system is hence determined by the state-space transformations:
\begin{equation} \label{eq:QRC1}
\begin{cases}
&A_t=T(A_{t-1},{\bf z}_t),\\
&\textbf{y}_t=h(A_t),
\end{cases}
\end{equation} 
where $t\in\bb{Z}$ denotes the time index. Analogously, one can define the same setting for semi-infinite discrete-time sequences: $(\bb{R}^n)^{\bb{Z}_-}=\{\underline{\textbf{z}}=(\dots,{\bf z}_{-1},{\bf z}_0)\ |\ {\bf z}_i\in \bb{R}^n, i \in \bb{Z}_- \}$ for left-infinite sequences and $(\bb{R}^n)^{\bb{Z}_+}=\{\underline{\textbf{z}}=({\bf z}_0,{\bf z}_1,\dots)\ |\ {\bf z}_i\in \bb{R}^n, i \in \bb{Z}_+ \}$ for right-infinite sequences. Similar definitions apply to $(D_n)^{\bb{Z}}$, $(D_n)^{\bb{Z}_-}$, and $(D_n)^{\bb{Z}_+}$ with elements in the subset $D_n\subset \bb{R}^n$. We can also construct sequence spaces $(\mathcal{B}(\mathcal{H}))^{\bb{Z}} $ for the space of bounded (trace-class) operators:
\begin{equation}
(\mathcal{B}(\mathcal{H}))^{\bb{Z}}= \{\underline{\textbf{A}}=(\dots,A_{-1},A_0,A_1,\dots)\\
\mid A_i\in\mathcal{B}(\mathcal{H}),\ i\in\bb{Z}\}. 
\end{equation}
Analogous definitions for $(\mathcal{B}(\mathcal{H}))^{\bb{Z}_-}$,$(\mathcal{B}(\mathcal{H}))^{\bb{Z}_+}$, $(\mathcal{S}(\mathcal{H}))^{\bb{Z}}$, $(\mathcal{S}(\mathcal{H}))^{\bb{Z}_-}$, and $(\mathcal{S}(\mathcal{H}))^{\bb{Z}_+}$  follow immediately.

The by-design CPTP character of the map $T:\mathcal{B}(\mathcal{H})\times \bb{R}^n\rightarrow \mathcal{B}(\mathcal{H})$ implies that it naturally restricts to a state equation $T:\mathcal{S}(\mathcal{H})\times \bb{R}^n\rightarrow \mathcal{S}(\mathcal{H})$ with density matrices as state space, that we shall use interchangeably in the sequel and denote using the same symbol.

\medskip

\subsection{Echo State Property (ESP)}
Consider the QRC system defined in~\eqref{eq:QRC1} or its analog for the subsets $\mathcal{S}(\mathcal{H})\subset \mathcal{B}(\mathcal{H})$ and $D_n\subset \bb{R}^n$, that is,  $T:\mathcal{S}(\mathcal{H})\times D_n\rightarrow \mathcal{S}(\mathcal{H})$. Given an input sequence $\underline{\textbf{z}} \in (D_n)^{\bb{Z}}$, we say that $\underline{\boldsymbol{\rho}} \in  (\mathcal{S}(\mathcal{H}))^{\bb{Z}} $ is a {\it solution} of~\eqref{eq:QRC1} for the input $\underline{\textbf{z}} $ if the  components of the sequences $\underline{\textbf{z}} $ and $\underline{\boldsymbol{\rho}}$ satisfy the first relation in~\eqref{eq:QRC1} for any $t \in \Bbb Z$.  We say that the QRC system has the {\it echo state property} (ESP) when it has a unique solution for each input $\underline{\textbf{z}} \in (D_n)^{\bb{Z}}$. More explicitly, for each $\underline{\textbf{z}}\in (D_n)^{\bb{Z}}$, there exists a unique sequence $\underline{\boldsymbol{\rho}}\in (\mathcal{S}(\mathcal{H}))^{\bb{Z}}$ such that 
\begin{equation}
\label{eq:QRC2}
\rho_t=T(\rho_{t-1},{\bf z}_t), \ \text{for all} \ t\in\bb{Z}.
\end{equation}

\medskip

\subsection{Filters and functionals} 
Let $\mathcal{S}(\mathcal{H})\subset \mathcal{B}(\mathcal{H})$ be the space of density matrices and let $D_n\subset \bb{R}^n$ be a subset in the input space. A map of the type $U: (D_n)^{\bb{Z}} \rightarrow (\mathcal{S}(\mathcal{H}))^{\bb{Z}}$ is called a {\it filter} associated to the QRC system~\eqref{eq:QRC1} when it satisfies that 
\begin{equation*}
U(\underline{\textbf{z}}) _t=T \left(U(\underline{\textbf{z}})_{t-1}, {\bf z} _t\right),\ \mbox{for all $\underline{\textbf{z}} \in (D_n)^{\bb{Z}}$ and $t \in \Bbb Z $.}
\end{equation*}
Filters induce what we call {\it functionals} $H: (D_n)^{\bb{Z}} \rightarrow \mathcal{S}(\mathcal{H})$ via the relation $H(\underline{\textbf{z}})=U(\underline{\textbf{z}}) _0 $. It is clear that a uniquely determined filter can be associated with a QRC system that satisfies the ESP. The filter maps, in that case, any input sequence to the unique solution of the QRC system associated with it. A filter is called {\it causal} if it only produces outputs that depend on present and past inputs. More formally, causality means that for any two inputs $\underline{\textbf{z}},\underline{\textbf{v}}\in (D_n)^{\bb{Z}}$ that satisfy ${\bf z}_{\tau}=\textbf{v}_{\tau}$ for any $\tau\leq t$, for a given $t\in\bb{Z}$, we have $U(\underline{\textbf{z}})_t=U(\underline{\textbf{v}})_t$. The filter $U$ is called {\it time-invariant} if there is no explicit time dependence on the system that determines it, that is, it commutes with the {\it time delay operator} defined as $\mathcal{T}_{\tau}(\underline{\textbf{z}})_t:={\bf z}_{t-\tau}$. Filters associated to QRC systems of the type~\eqref{eq:QRC1} are always causal and time-invariant (Proposition 2.1 in Ref.~\cite{grigoryeva2018echo}). As noted in previous works \cite{boyd1985fading, grigoryeva2018echo, grigoryeva2018universal}, there is a bijection between causal and time-invariant filters and functionals on $(D_n)^{\bb{Z}_-}$.  Then, we can restrict our work to causal and time-invariant filters with target and domain in spaces of left-infinite sequences. 

\medskip

\subsection{Fading Memory Property (FMP)}
Let $w:\bb{N}\rightarrow (0,1]$ be a decreasing sequence with zero limit and $w_0=1$. The weighted norm $||\cdot||_w$ on $(\bb{R})^{\bb{Z_-}}$ is defined as 
\begin{equation}
||\underline{\textbf{z}}||_w:=\sup_{t\in \bb{Z}_-}\{w_{-t}||{\bf z}_t||\},
\end{equation}
and the space 
\begin{equation}
l^{w}_-(\bb{R}^n)=\{\underline{\textbf{z}}\in (\bb{R}^n)^{\bb{Z_-}}|\ ||\underline{\textbf{z}}||_w<\infty\},
\end{equation}
with weighted norm $||\cdot||_w$ forms a Banach space (see Appendix A.2 in Ref.~\cite{grigoryeva2018echo}).
In the same vein, we can define 
\begin{equation}
\begin{split}
&||\underline{\textbf{A}}||_w:=\sup_{t\in \bb{Z}_-}\{w_{-t}||A_t||\}, \\
&l^{w}_-(\mathcal{B}(\mathcal{H}))=\{\underline{\textbf{A}}\in (\mathcal{B}(\mathcal{H}))^{\bb{Z_-}}|\ ||\underline{\textbf{A}}||_w<\infty\}.
\end{split}
\end{equation}
It can be shown that $l^{w}_-(\mathcal{B}(\mathcal{H}))$ is a Banach space as well. As explained in Ref.~\cite{martinez2023quantum}, $\mathcal{S}(\mathcal{H}) $ is necessarily compact. An important consequence of this fact is that the relative topology induced by the $l^{w}_-(\mathcal{B}(\mathcal{H})) $ on  $(\mathcal{S}(\mathcal{H}))^{\bb{Z}_-} $ coincides with the product topology (see Corollary 2.7 in Ref.~ \cite{grigoryeva2018echo}).

Take now a subset $D_n\subset \bb{R}^n$ such that  $(D_n)^{\bb{Z}_-}\subset l^{w}_-(\bb{R}^n)$ and consider a QRC system $T:\mathcal{S}(\mathcal{H})\times D_n\rightarrow \mathcal{S}(\mathcal{H})$ that has the ESP. We say that $T$ has the {\it fading memory property} (FMP) when the corresponding functional $H:(D_n)^{\bb{Z}_-} \rightarrow\mathcal{S}(\mathcal{H})$ is a continuous map between the metric spaces $((D_n)^{\bb{Z}_-},||\cdot||_w)$ and $((\mathcal{S}(\mathcal{H}))^{\bb{Z}_-},||\cdot||_w)$, for some weighting sequence $w$. If  $D_n$ is compact, once $H$ is continuous for a given weighting sequence $w$, then it is continuous for all weighting sequences (see Ref.~\cite[Theorem 2.6]{grigoryeva2018echo}).

Proposition 3 of Ref.~\cite{martinez2023quantum} contains a necessary and sufficient contractivity condition for the ESP and the FMP to hold in the case of QRC models with compact input spaces. However, one could still find that despite the ESP and the FMP, the filter of the quantum reservoir is trivially input-independent. To make this observation more specific, let us define, for any ${\bf z}\in D_n$, the map $T_{\bf z}:=T(\cdot , {\bf z}):\mathcal{B}(\mathcal{H}) \rightarrow \mathcal{B}(\mathcal{H})$ by fixing the input. The contractivity condition in Proposition 3 in Ref.~\cite{martinez2023quantum} that characterizes the ESP and the FMP amounts to saying that the maps $T_{\bf z}(\cdot)$ are all contractions and hence have a unique fixed point $\rho^*({\bf z})\in D _n $. Denote by $\rho^*: D_n \rightarrow \mathcal{S}(\mathcal{H})$ the map that associates to each input value ${\bf z}\in D_n$ the fixed point $\rho^*({\bf z})\in D _n $ of the corresponding map $T_{\bf z} $. We call $\rho^*: D_n \rightarrow \mathcal{S}(\mathcal{H})$ the {\it fixed point function}. Theorem 2 in Ref.~\cite{martinez2023quantum}, which we reproduce below for the sake of completeness,  shows that a quantum reservoir filter is trivial if and only if the corresponding fixed point function is constant. 

\begin{theorem}\label{th:1}
Let $T: \mathcal{B}(\mathcal{H}) \times D_n\rightarrow \mathcal{B}(\mathcal{H})$ be a QRC system for which there exists an operator norm $\vertiii{\cdot } $ and $\epsilon>0 $ such that 
$
\vertiii{T(\cdot  , {\bf z})|_{\mathcal{B}_0(\mathcal{H})}}<1- \epsilon$,  for all ${\bf z} \in D_n$. Then, $T$ has a constant fixed point  map $\rho^*: D_n \rightarrow \mathcal{S}(\mathcal{H})$ that always takes the value $\rho^*_T $, that is, $T(\rho^*_T,{\bf z})=\rho^*_T$, for all ${\bf z}\in D_n$, if and only if the corresponding filter $U _T  $ is constant, that is, $U_T(\underline{{\bf z}})_t=\rho^\ast_T  \in \mathcal{S} ({\mathcal H})$, for all $\underline{{\bf z}}\in (D_n)^{\mathbb{Z}}$ and $t \in \Bbb Z$.
\end{theorem} 

The symbol $\mathcal{B}_0(\mathcal{H})\subset \mathcal{B}(\mathcal{H})$ denotes the vector subspace made of traceless operators, and the operator norm contraction, i.e. $
\vertiii{T(\cdot  , {\bf z})|_{\mathcal{B}_0(\mathcal{H})}}<1- \epsilon$ for all ${\bf z}\in D_n$, is a necessary and sufficient condition for the ESP and FMP when $D_n$ is compact (Proposition 3 in Ref.~\cite{martinez2023quantum}). The $\epsilon$ in the contraction condition ensures that the filter can be expressed as a uniformly convergent series [see Eq.~\eqref{eq:filter_x}] later on in the text.

\begin{remark}\normalfont\label{remark:1}
Theorem \ref{th:1} can be directly generalized to any state-space map that is continuous and a contraction on the first entry. That is, let $F: V \times D_n\rightarrow V$, assume 
\begin{equation*}
||F({\bf x}^2,{\bf z})-F({\bf x}^1,{\bf z})||\leq c ||{\bf x}^2-{\bf x}^1||,
\end{equation*}
for all ${\bf x}^1,{\bf x}^2\in V \text{ and } {\bf z}\in D_n$, with constant $0<c<1$ for some norm (we assume that $V$ is a closed subset of a Banach space). Let us define the maps $F_{\bf z}:V \rightarrow V$, ${\bf z}\in D_n$, by fixing the input. Then, by the Banach-fixed point theorem, the maps $F_{\bf z}$ have a unique fixed point that can be used to define a fixed point function ${\bf x}^*:D_n\rightarrow V$. This fixed point map is constant and equal to a given value ${\bf x}^*_F$, that is, $F({\bf x}^*_F,{\bf z})={\bf x}^*_F$, for all ${\bf z}\in D_n$, if and only if the corresponding filter $U _F$ is constant, that is, $U_F(\underline{{\bf z}})_t={\bf x}^\ast_F  \in V$, for all $\underline{{\bf z}}\in (D_n)^{\mathbb{Z}}$ and $t \in \Bbb Z $. The proof mimics that of Theorem 2 in Ref.~\cite{martinez2023quantum}.
\end{remark}

\subsection{SAS representation} \label{sec:SAS_definitions}
Theorem \ref{th:1} can be shown either in the density matrix formalism or in the matrix representation.  More specifically, we shall be working with the Bloch vector (or Pauli matrix) representation of quantum finite-dimensional systems \cite{kimura2003bloch,byrd2003characterization,hantzko2024pauli} associated with a given Gell-Mann basis.
We will start by introducing the notation necessary for the matrix representation of quantum channels. Let $\{B_i\}_{i \in \left\{1, \ldots, d ^2\right\}}$ be an orthonormal basis for the vector space $\mathcal{B}(\mathcal{H})$, when endowed with the Hilbert-Schmidt inner product, that is, $\tr(B^{\dagger}_iB_j)=\delta_{ij}$, and $d\in \mathbb{N}$ is the dimension of $\mathcal{H}$. Using any such basis we can represent any operator  $A \in \mathcal{B}(\mathcal{H})$ as $A=\sum_{i=1}^{d^2} a_iB_i$, with $a_i=\tr(B_i^{\dagger}A)$. Analogously, we can express any linear map $T: \mathcal{B}(\mathcal{H})\rightarrow \mathcal{B}(\mathcal{H})$ as
\begin{equation}
\label{eq:Wmatrix}
T(A)
=\sum^{d^2}_{i,j=1}\widehat{T}_{ij}a_jB_i, \ \mbox{where $\widehat{T}_{ij}=\tr(B^{\dagger}_iT(B_j))$.}
\end{equation}
This observation implies that QRC systems $T:\mathcal{S}(\mathcal{H})\times D_n\rightarrow \mathcal{S}(\mathcal{H})$ admit an equivalent representation as a system 
$\widehat{T}: V \times D_n\rightarrow V$, where $V\subset \bb{R}^{d^2}$ is the subset of real Euclidean space that contains the coordinate representations of the elements in $\mathcal{S}(\mathcal{H}) $ using the basis $\{B_i\}_{i \in \left\{1, \ldots, d ^2\right\}}$. 

This statement has been formalized in Ref.~\cite{martinez2023quantum} using the language of system morphisms (see Refs.~\cite{RC15, RC16} for the standard definitions and elementary facts). 

Indeed, consider the QRC system given by the quantum channel $T:\mathcal{S}(\mathcal{H})\times D _n\rightarrow \mathcal{S}(\mathcal{H})$ and the readout map $h:\mathcal{B}(\mathcal{H})\rightarrow \bb{R}^m$, $m\in \bb{N}$. Using the orthonormal basis $\mathcal{B}=\{B_i\}_{i \in \left\{1, \ldots, d ^2\right\}}$ and the discussion above define the map
\begin{equation}
\label{system isom linear}
\begin{array}{cccc}
G_{\mathcal{B}}: &\mathbb{C}^{d^2} &\longrightarrow & \mathcal{B}({\mathcal H})\\
&\mathbf{a} &\longmapsto &\sum_{i=1}^{d^2} a _i B _i.
\end{array}
\end{equation}

Define $V=G_{\mathcal{B}} ^{-1}(\mathcal{S}(\mathcal{H}))\subset \bb{R}^{d^2}$ as well as the map (that we denote with the same symbol) $G_{\mathcal{B}}:V \longrightarrow \mathcal{S}(\mathcal{H}) $  that we obtain by restriction of the domain and codomain in \eqref{system isom linear}. This restricted map is a homeomorphism when $V$ and $\mathcal{S}(\mathcal{H}) $ are endowed with their relative topologies \cite[Theorem 18.2]{munkres2000topology}. With all these ingredients, it is straightforward to verify that the QRC system $(\mathcal{S}(\mathcal{H}), T, h) $ is system isomorphic to $(V, \widehat{T}, \widehat{h}) $ with $\widehat{T}:V \times D_n \longrightarrow V $  and $\widehat{h}: V \longrightarrow \mathbb{R}^m $ given by
\begin{align}
\widehat{T}(\mathbf{a}, {\bf z})&:= G_{\mathcal{B}}^{-1} \left(T \left(G_{\mathcal{B}}(\mathbf{a}), {\bf z}\right)\right), \label{isomorphic state map qrc}\\
\widehat{h}(\mathbf{a})&:= h(G_{\mathcal{B}}(\mathbf{a})),\label{isomorphic readout map qrc}
\end{align}
and that the isomorphism is given by the map $G_{\mathcal{B}}:V \longrightarrow \mathcal{S}(\mathcal{H}) $. The procedure that we just spelled out can be reproduced for any other (orthonormal) basis $\mathcal{B} '$ of $\mathcal{B}({\mathcal H}) $, in which case we would obtain another system $(V', \widehat{T}', \widehat{h}') $ which is obviously isomorphic to both $(V, \widehat{T}, \widehat{h}) $ and $(\mathcal{S}(\mathcal{H}), T, h) $.

As it has been proved in Proposition 2 of Ref.~\cite{martinez2023quantum}, the importance of these statements lies in the fact that system isomorphisms preserve the ESP and the FMP and create a straightforward relation between the corresponding filters. More specifically, the filters $U _T $ and $U_{\widehat{T}}$ (respectively,  $U _T^h$ and  $U_{\widehat{T}}^{\widehat{h}} $) determined by $T$  and $\widehat{T} $ (respectively, by $(T,h)$  and $(\widehat{T}, \widehat{h}) $) satisfy that $U _T={\cal G}_{\mathcal{B}} \circ U_{\widehat{T}} $ (respectively, $U _T^h= U_{\widehat{T}}^{\widehat{h}} $), where ${\cal G}_{\mathcal{B}}=\prod _{\Bbb Z} G _{\mathcal{B}}: (V)^{\mathbb{Z}}\longrightarrow \left({\cal S}( {\mathcal H})\right)^{\mathbb{Z}}$ stands for the product map.  

We now implement these ideas by choosing a generalized Gell-Mann basis \cite{kimura2003bloch,byrd2003characterization,siewert2022orthogonal,hantzko2024pauli} of ${\cal B}({\mathcal H}) $. This is an orthonormal basis made of Hermitian operators in which, by convention, its first element is the normalized identity, namely, $B_1=I/\sqrt{d}$. The remaining $(d^2-1)$ traceless Hermitian operators are the generators 
\begin{equation}
\label{zero trace gens}
\mathcal{B}_0=\left\{\sigma_i\right\}_{i \in \left\{1, \ldots, d^2-1\right\}}
\end{equation}
of the {\it fundamental representation}   of the Lie algebra $\mathfrak{su}(d)$ of SU($d$). The case  $d=2$ corresponds to the case of one qubit, and the Gell-Mann basis is made of the standard Pauli matrices. The orthonormality of the Gell-Mann basis is guaranteed by the product property of the fundamental representation of  $\mathfrak{su} (d) $, namely, $\sigma_a\sigma_b=\delta_{ab}I/(2d)+\sum^{d^2-1}_{c=1}f_c\sigma_c$, where $\sigma_a, \sigma_b $ are two elements of the Gell-Mann basis for $\mathfrak{su} (d) $,  and $f_c$ are complex coefficients. The resulting Gell-Mann basis $\mathcal{B}=\left\{B_i\right\}_{i \in \left\{1, \ldots, d^2\right\}}$  of ${\cal B}({\mathcal H}) $ is hence given by $B _1=I/\sqrt{d}$ and $B _i= \sigma_{i-1} $, $1<i\leq d ^2 $. Note that the subset $\mathcal{B} _0=\left\{B_i\right\}_{i \in \left\{2, \ldots, d^2\right\}}$ is a basis for the vector subspace ${\cal B} _0({\mathcal H}) \subset {\cal B}({\mathcal H}) $ of codimension $1$ made of traceless operators.

We now spell out the matrix expression $\widehat{T} :\mathbb{C}^{d^2} \times D _n \longrightarrow \mathbb{C}^{d^2} $ of a CPTP map $T: \mathcal{B}({\mathcal H})\times D_n\longrightarrow \mathcal{B}({\mathcal H})$ in the basis $\mathcal{B}  $  that we just introduced, by using the prescription introduced in \eqref{eq:Wmatrix}. 
We first note that the choice $B_1=I/\sqrt{d}$ and the trace-preserving character of $T(\cdot , {\bf z})$ for any ${\bf z} \in D_n $, imply that  $\widehat{T}(\cdot , {\bf z})_{11}  =1$. Analogously, $\widehat{T}(\cdot , {\bf z})_{1j}=\tr(T(B_j, {\bf z}))/\sqrt{d}=0$ for $1<j\leq d^2$, since $T$ is trace-preserving. This implies that the matrix $\widehat{T}(\cdot , {\bf z})$ can be written as
\begin{equation}
\label{matrix form t hat}
\widehat{T}(\cdot , {\bf z})=\left(\begin{matrix}
1 & {\bf 0}_{d^2-1}  \\
q({\bf z}) & p({\bf z})  
\end{matrix}\right),
\end{equation}
where $p({\bf z}) $ is the square matrix of dimension $d^2-1 $ with complex entries 
\begin{multline}\label{eq:p_ij}
p({\bf z})_{ij}:=\left(\widehat{T}(\cdot  , {\bf z})|_{G_{\mathcal{B}}^{-1} \left(\mathcal{B}_0(\mathcal{H})\right)}\right)_{ij}\\
=
\left(\widehat{T}(\cdot  , {\bf z})|_{G _{\mathcal{B}} ^{-1}\left(\operatorname{span} \left\{\mathcal{B} _0\right\}\right)}\right)_{ij}
=\tr(B_iT(B_j, {\bf z})),
\end{multline}
$1<i,j\leq d^2$, and $q({\bf z})\in \mathbb{C}^{d^2-1} $ is given by $q({\bf z})_{i}=\tr(B_iT(I, {\bf z}))/\sqrt{d}$, $1<i\leq d^2$. The symbol ${\bf 0}_{d^2-1}$ denotes a zero-row of length $d^2-1$.

Now, given that any element $\rho \in \mathcal{S}({\mathcal H}) $ has coordinates in the Gell-Mann basis of the form $\left(1/\sqrt{d}, \mathbf{x}^{\top} \right)^{\top}\in \mathbb{C}^{d^2}   $ with $\mathbf{x} \in \mathbb{C}^{d^2-1} $ called the {\it Bloch vector}, the matrix form \eqref{matrix form t hat} implies that 
\begin{equation}
\label{matrix form t hat with vector}
\widehat{T}(\cdot , {\bf z})\left(\begin{matrix}
1/\sqrt{d}  \\
\mathbf{x}  
\end{matrix}\right)=
\left(\begin{matrix}
1/\sqrt{d}  \\
p({\bf z})\textbf{x}+q({\bf z})  
\end{matrix}\right).
\end{equation}
This expression implies that  $\widehat{T} :V \times D _n \longrightarrow V $ admits a system-isomorphic representation $\widehat{T}_0 :V_0 \times D _n \longrightarrow V_0 $ on the set $V _0= \left\{\mathbf{x} \in \mathbb{C}^{d^2-1}\mid \left(1/\sqrt{d}, \mathbf{x}^{\top} \right)^{\top}\in  V\right\}\subset \mathbb{C}^{d^2-1}$ given by 
\begin{equation}
\label{eq:x}
\begin{array}{cccc}
\widehat{T} _0: & V _0 \times D_n & \longrightarrow & V _0,\\
&(\mathbf{x}, {\bf z}) & \longmapsto & p({\bf z})\textbf{x}+q({\bf z}),
\end{array}
\end{equation}
with readout $\widehat{h} _0(\mathbf{x})= \widehat{h} \left(\left(1/\sqrt{d}, \mathbf{x}^{\top} \right)^{\top}\right)$. The system isomorphism  is in this case, given by the map 
\begin{equation}
\label{isomor v vo}
\begin{array}{cccc}
i _0: &V _0  &\longrightarrow &V, \\
&\mathbf{x} \in V _0 &\longmapsto &\left(1/\sqrt{d}, \mathbf{x}^{\top} \right)^{\top}. 
\end{array}
\end{equation}
Part (iv) of Proposition 2 in Ref.~\cite{martinez2023quantum} guarantees that the dynamical properties of the system $(\widehat{T}, \widehat{h}, V) $ (and hence those of the QRC system $(T,h, \mathcal{S}({\mathcal H}))$)
are  equivalent to those of $(\widehat{T}_0, \widehat{h}_0, V_0) $.

The importance of this observation is that it links $(T,h, \mathcal{S}({\mathcal H}))$ to the nonhomogeneous state-affine system (SAS) introduced in Refs.~\cite{Sontag1979, DangVanMien1984} and whose universality has been proved in Ref.~\cite{grigoryeva2018universal}. SAS are defined as state equations that have the form spelled out in \eqref{eq:x}. Strictly speaking, the SAS systems studied in the above-cited references impose polynomial or trigonometric dependences of $q$ and $p$ on the inputs, while in our situation is capable of accommodating more general forms.

It has been shown in Ref.~\cite{grigoryeva2018universal} that when such a system has the ESP, the corresponding filter $ U_{\widehat{T}_0}: \left(D_n\right)^{\mathbb{Z}} \rightarrow  \left(V _0\right)^{\mathbb{Z}}$ can be written as
\begin{equation} 
\label{eq:filter_x}
U_{\widehat{T}_0}(\underline{{\bf z}})_t = \sum^{\infty}_{j=0}\left(\prod^{j-1}_{k=0}p({\bf z}_{t-k})\right)q({\bf z}_{t-j}),
\end{equation}
where 
$\prod_{k=0}^{j-1}p({\bf z}_{t-k}):=p({\bf z}_{t}) \cdot p({\bf z}_{t-1}) \cdots p({\bf z}_{t-j+1})$.

\section{Results}\label{sec:results}
\subsection{Injective filters}

The situation presented in Theorem \ref{th:1} describes reservoir design features that should be avoided, as they produce systems with only trivial solutions. More in detail, using the SAS representation in the previous section, it was shown in Theorems 1 in Ref.~\cite{martinez2023quantum} that contractive quantum channels $T: \mathcal{B}(\mathcal{H}) \times D_n\rightarrow \mathcal{B}(\mathcal{H})$ produce the trivial filter, that is, $U_{\widehat{T}_0}(\underline{{\bf z}})_t ={\bf 0}$ for all $\underline{{\bf z}} \in \left(D_n\right)^{\mathbb{Z}} $, if and only if they are {\it unital}, that is they satisfy $T(I,{\bf z})=I$, for all ${\bf z} \in D_n$. Even though it is clear that this feature should be avoided, this is not the end of the design question as we could still have situations where the ESP and the FMP hold and the filter is nontrivially input-dependent for long sequences, but {\it the model does not differentiate between some of those input sequences}. To prevent this from happening, we shall formulate conditions on QRC systems that guarantee that the corresponding filters are {\it injective}, that is,  they map distinct input sequences to distinct output sequences. Equivalently, given a filter $U: (D_n)^{\mathbb{Z}_-} \rightarrow V^{\mathbb{Z}_-} $, we say that  $U$ is {\it injective} when given $\underline{{\bf z}},\underline{{\bf z}}'\in (D_n)^{\mathbb{Z}_-}$  $U(\underline{{\bf z}})= U(\underline{{\bf z}}')$, we necessarily have that $\underline{{\bf z}}=\underline{{\bf z}}'$.

In this section we introduce a condition that ensures filter injectivity and later on in Proposition \ref{prop:1} we formulate a readily verifiable sufficient condition for it to hold in the QRC context.

\begin{definition}
\label{siinvert def}
A reservoir system $F:V\times D_n \rightarrow V$ is called {\it state-input invertible (SI-invertible)} when the maps $F_{\bf x}:D_n\rightarrow V$ given by $F_{\bf x}({\bf z}):=F({\bf x},{\bf z})$ are injective for any ${\bf x}\in V$. 
\end{definition}

Notice that in the presence of SI-invertibility, for any $\textbf{x}\in V$ there exists a set $V_\textbf{x}:=F_\textbf{x}(D_n)$ such that the map $F_\textbf{x}:D_n\rightarrow V_\textbf{x}\subset V$ has an inverse $F^{-1}_\textbf{x}:V_\textbf{x}\rightarrow D_n$. The notion of state-invertibility has been introduced in Ref.~\cite{manjunath2021universal} for the construction of embedding techniques in the forecasting of deterministic dynamical systems. The following lemma establishes the link between SI-invertibility and filter injectivity.

\begin{lemma}\label{lm:1}
If $F:V\times D_n \rightarrow V$ is SI-invertible and it has the ESP, then the filter $U_F:(D_n)^{\mathbb{Z}_-}\rightarrow (V)^{\mathbb{Z}_-}$ is necessarily injective.
\end{lemma} 
\begin{proof}
Suppose that $U_F$ is not injective. This means that there exists $\underline{\textbf{z}},\underline{\textbf{z}}'\in (D_n)^{\mathbb{Z}_-}$ such that $\underline{\textbf{z}}\neq \underline{\textbf{z}}'$ but $U_F(\underline{\textbf{z}})= U_F(\underline{\textbf{z}}')= \underline{\textbf{x}}\in (V)^{\mathbb{Z}_-}$. This implies that for any $t\in \mathbb{Z}_- $, $F(\textbf{x}_{t-1},\textbf{z}_t) = F(\textbf{x}_{t-1},\textbf{z}'_t)$ or, equivalently, $F_{\textbf{x}_{t-1}}(\textbf{z}_t)=F_{\textbf{x}_{t-1}}(\textbf{z}'_t)$, which implies that $\textbf{z}_t = \textbf{z}'_t$ for all  $t\in \mathbb{Z}_-$,  and hence
$\underline{\textbf{z}}=\underline{\textbf{z}}'$.
\end{proof}

One way to guarantee SI-invertibility is using the Local Injectivity Theorem \cite[Theorem 2.5.10]{mta} by requiring that the linear maps
\begin{equation}
\label{maximal rank cond}
\mathcal{D}F_\textbf{x}(\textbf{z}):\mathbb{R}^n\rightarrow\mathbb{R}^N \  \mbox{have rank $n$ for all $\textbf{x}\in V$,  $\textbf{z}\in D_n$.} 
\end{equation}
The symbol $\mathcal{D}F_\textbf{x}$ denotes the Fr\'echet differential of $F_\textbf{x}$. Note that the rank condition in \eqref{maximal rank cond} can hold obviously only when $n\leq N $. Notice, additionally, that the injectivity for all $\textbf{x}\in V$ is not needed to guarantee SI-invertibility, and it suffices to require it as a condition exclusively on the subset $V_R\subset V$ of {\it reachable states} of the system defined by
\begin{equation}
\label{reachable set}
V_R := \{{\bf x}\in V \ |\ {\bf x}=U_{F}(\underline{{\bf z}})_0 \text{ for some } \underline{{\bf z}}\in (D_n)^{\bb{Z}_-} \}.
\end{equation}

We now concentrate on how filter injectivity can be guaranteed in the context of the SAS family using these observations. All the results will be formulated for these models in either classical or quantum contexts, while the examples will be entirely quantum-related. When $F(\textbf{x},\textbf{z})=p(\textbf{z})\textbf{x}+q(\textbf{z})$ with $\textbf{z}\in D_n \subset \mathbb{R}^n$ and $\textbf{x}\in V\subset \mathbb{R}^N$, then for any $\textbf{v}\in {\Bbb R}^n$, $\mathcal{D}F_\textbf{x}(\textbf{z})(\textbf{v})=\left(\mathcal{D}p(\textbf{z})(\textbf{v})\right)\textbf{x}+\mathcal{D}q(\textbf{z})(\textbf{v})$ with $\mathcal{D}p(\textbf{z}):\mathbb{R}^n\rightarrow \mathbb{R}^{N\times N}$ and $\mathcal{D}q(\textbf{z}):\mathbb{R}^n\rightarrow \mathbb{R}^N$, the corresponding Fr\'echet derivatives. Consequently, using the condition \eqref{maximal rank cond}, if the linear maps 
$$\left(\mathcal{D}p(\textbf{z})(\cdot)\right)\textbf{x}+\mathcal{D}q(\textbf{z}):\mathbb{R}^n\rightarrow \mathbb{R}^N$$
have rank $n$  for all $\textbf{x}\in V_R$ and $\textbf{z}\in D_n$ then $F$ is necessarily SI-invertible. 
These observations, together with Lemma \ref{lm:1}, allow us to formulate the following statement in the SAS context:

\begin{proposition}\label{prop:1}
Let $F=V\times D_n \rightarrow V$ be a differentiable SAS system given by $F(\mathbf{x}, {\bf z})=p({\bf z}){\bf x}+q({\bf z})$ and let us assume that it has the ESP. Let $V_R\subset V$ be the subset of reachable states of the system. Suppose that for all ${\bf x}\in V_R$ and ${\bf z}\in D_n$ the linear maps 
\begin{equation}
\label{rank condition}
\left(\mathcal{D}p({\bf z})(\cdot)\right){\bf x}+\mathcal{D}q({\bf z}):\mathbb{R}^n\rightarrow \mathbb{R}^N
\end{equation}
have all rank $n$ and hence are injective. Then, the corresponding filter $U_F:(D_n)^{\mathbb{Z}_-}\rightarrow (V)^{\mathbb{Z}_-}$ is injective.
\end{proposition}

\begin{remark}
\normalfont
If the input space $D_n  $ is compact, the ESP hypothesis in the previous statement implies, according to Proposition 3 in Ref.~\cite{martinez2023quantum}, that $p $ is necessarily contractive, that is, there is a matrix norm $\vertiii{\cdot }$ in the space of $N\times  N $ matrices such that $
\vertiii{p({\bf z})}<1- \epsilon$, for all ${\bf z} \in D_n$, and some $\epsilon >0 $, in which case the corresponding filter is given by~\eqref{eq:filter_x}. 
\end{remark}

\begin{remark}
\normalfont
The approach to filter injectivity that has been tackled in Proposition \ref{prop:1} hinges on the notion of state-input invertibility introduced in Definition \ref{siinvert def} and on Lemma \ref{lm:1}. An alternative approach to the same problem can be taken by applying the inverse function theorem to the Fr\'echet differential of the filter \eqref{eq:filter_x} along the lines introduced in Corollary 23 in Ref.~\cite{RC9}. An implementation of that result for SAS systems leads to the same criterion on the linear map \eqref{rank condition} that was formulated in Proposition \ref{prop:1}.
\end{remark}

\begin{remark}\label{remark:8}
\normalfont
{\it  Dynamical systems learnability and filter injectivity.} One of the most visible applications of RC is the learning and forecasting of deterministic dynamical systems \cite{Jaeger04, Ott2018, Pathak:PRL, pathak:chaos}. In that setup, given a reservoir system $F: \mathbb{R}^N \times  \mathbb{R} \times \mathbb{R}^N $ with the ESP and an invertible discrete-time dynamical system $\phi \in {\rm Diff}(M) $ on a manifold $M$, the task consists of forecasting the time series $S_{\omega, \phi}(m):= \left(\omega(\phi ^t(m))\right)_{t \in \Bbb Z}$, $m \in M $, given by the one-dimensional observations of the dynamical system associated with the observation map $\omega:M \rightarrow \mathbb{R}$ and $m\in M$. In that setup, this time series is said to be {\it learnable} with the reservoir system $F$ when being fed to it up to time $t \in \Bbb Z $, its term at time $t + 1  $ can be written (or forecasted) out of the reservoir state at time $t$ by using a deterministic readout $h: \mathbb{R} ^N \rightarrow \mathbb{R}$. More precisely, 
\begin{equation}
\label{forecasting equation}
\omega(\phi ^{t+1}(m))=h \left(U _F\left(S_{\omega, \phi}(m)\right)_t\right),\  \mbox{for all $t \in \Bbb Z $},
\end{equation}
and some readout $h: \mathbb{R} ^N \rightarrow \mathbb{R}$. An important situation in which learnability holds is in the presence of {\it injective generalized synchronizations} \cite{lu:bassett:2020, Verzelli2020b, grigoryeva2021chaos, RC21}. Indeed, consider the map $f:M \rightarrow \mathbb{R}^N  $ given by $f(m):= U _F\left(S_{\omega, \phi}(m)\right)_0 $, $m \in M $. If this map is injective, the forecasting strategy in \eqref{forecasting equation} can be implemented by setting $h:= \omega\circ \phi \circ f ^{-1} $. As it has been shown in Ref.~\cite{RC21}, dynamical systems learning using injective generalized synchronizations is a strict generalization of the embedology techniques \cite{sauer1991embedology} based on the celebrated Takens delay embedding theorem \cite{takensembedding}.

The connection between the filter injectivity condition that we study in this paper and the dynamical systems learning using injective generalized synchronizations lies in the fact that {\it the filter $U _F $ of a reservoir system $F$ that produces an injective generalized synchronization $f $ is necessarily injective} on the image of the map $S_{\omega, \phi}:M \rightarrow \mathbb{R} ^{\mathbb{Z}_{-}} $. Indeed, let $m, m' \in M $ be such that $U_F \left(S_{\omega, \phi}(m)\right)=U_F \left(S_{\omega, \phi}(m')\right) $; this implies in particular that $f(m)=U_F \left(S_{\omega, \phi}(m)\right)_0=U_F \left(S_{\omega, \phi}(m')\right)_0=f(m') $ and since $f$ is by hypothesis injective, we have that $m=m' $, which implies that $S_{\omega, \phi}(m)=S_{\omega, \phi}(m') $, as required.
\end{remark}

\begin{example}
\normalfont
An example in which Proposition \ref{prop:1} necessarily fails are the cases where the rank of the linear maps  $\mathcal{D}F_\textbf{x}(\textbf{z})$ is zero. If this happens then $\mathcal{D}F_\textbf{x}(\textbf{z}) = 0$ for all $\mathbf{x} $  and ${\bf z}  $ and hence the corresponding filter is constant. In QRC models, the necessary and sufficient conditions for having a constant filter are given by Theorem \ref{th:1}, and it is easy to verify that under those conditions, $\mathcal{D}F_\textbf{x}(\textbf{z}) = 0$ necessarily. See, for instance, Examples 2 and 3 of Ref.~\cite{martinez2023quantum} in the Mathematica notebook in Ref.~\cite{github}.
\end{example}

\begin{example}
\normalfont
We now construct a QRC system that satisfies the injectivity condition in \eqref{rank condition} for all $\textbf{z}\in D_n$. Consider the QRC system  $Q:\mathcal{B}(\mathcal{H})\times D_n\rightarrow \mathcal{B}(\mathcal{H})$ and assume that for any ${\bf z}\in D_n $ the maps $Q(\cdot , {\bf z}) $ are {\it ergodic}, that is, they have a unique fixed point $\rho_Q^\ast ({\bf z}) \in \mathcal{S}({\mathcal H})$. In that case, it can be shown (see Corollary 2 in Ref.~\cite{burgarth2013ergodic}) that all the other fixed points $A\in \mathcal{B}(\mathcal{H})$ of $Q(\cdot , {\bf z}) $, that is,  $Q(A, {\bf z})=A$, are of the form $A=\tr(A)\rho_Q^*({\bf z})$. We now use $Q$ to construct another quantum channel $T:\mathcal{B}(\mathcal{H})\times D_n\rightarrow \mathcal{B}(\mathcal{H})$ given by
\begin{equation}
T(A,\textbf{z})=\rho_Q^*(\textbf{z})\tr(A),
\end{equation}
where $\rho_Q^*(\textbf{z})$ is the fixed point of the channel $Q$. It is easy to see using \eqref{eq:p_ij} that the SAS representation of this map is 
\begin{equation}
\hat{T}(\textbf{z})(\cdot):=\left(\begin{matrix}
1 & \boldsymbol{0}_{d^2-1}  \\
q(\textbf{z}) & \boldsymbol{0}_{d^2-1\times d^2-1}  
\end{matrix}\right),
\end{equation}
where $q({\bf z})\in \mathbb{C}^{d^2-1} $ is given by 
\begin{equation*}
q({\bf z})_{i}=\tr(B_iT(I, {\bf z}))/\sqrt{d}= \sqrt{d}\,\tr(B_i\rho_Q^*(\textbf{z})), \  1<i\leq d^2.
\end{equation*}
The associated filter  $U_{\hat{T}}: D_n^ {\mathbb{Z}_{-}}\rightarrow \left(\mathbb{C}^{d^2-1}\right)^ {\mathbb{Z}_{-}} $ is given by $U_{\hat{T}}(\underline{{\bf z}})_t= q({\bf z}_t)$, $t \in \mathbb{Z}_{-}$. According to Proposition \ref{prop:1}, the injectivity of this filter is guaranteed if the linear maps $\mathcal{D} q(\textbf{z}): \mathbb{R} ^n \rightarrow  \mathbb{C}^{d^2-1} $ whose components are given by
\begin{equation}
\mathcal{D}q({\bf z})_{i}= \sqrt{d}\,\tr(B_i\mathcal{D}\rho_Q^*(\textbf{z})),\  1<i\leq d^2,
\end{equation}
are injective for any ${\bf z} \in D _n $. This can be ensured by choosing an appropriate $\rho^*_Q(\textbf{z})$. Note that this example shows how the injectivity of the filter is, in principle, not related to potential good memory properties of the reservoir since the filter $U_{\hat{T}}(\underline{{\bf z}})_t= q({\bf z}_t)$, $t \in \mathbb{Z}_{-}$,  is memoryless as its output only depends on the input that is contemporaneously fed at each time step. 
\end{example}

The rank conditions \eqref{maximal rank cond} or \eqref{rank condition} in Proposition \ref{prop:1} might be difficult to evaluate for all inputs in an arbitrary QRC model. Or it could even happen that the filter is not globally injective but only on the neighborhood of a given input sequence. A local version of this result can be formulated that only requires verifying the condition at one point ${\bf x}_0\in V$, but in exchange, provides only local filter injectivity on a neighborhood of the input sequences that generate this output point. To characterize these input sequences, we will focus on constant output sequences $\underline{\bf x}_0 = (\dots,{\bf x}_0,{\bf x}_0)$. The Banach-fixed point theorem offers a natural way to obtain $\underline{\bf x}_0$ through a constant input sequence when $F$ is a contraction on the first entry. In such a case, $\underline{{\bf x}}_0 = U_F(\underline{\bf z})$ for $\underline{\bf z} = (\dots, {\bf z},{\bf z})$. However, constant input sequences are not necessarily the only route towards constant output sequences. The next proposition accounts for this possibility when studying local filter injectivity.

\begin{proposition}
\label{local proposition 1}
Let $F:V\times D_n \rightarrow V$ be a contractive differentiable reservoir map, let ${\bf x}_0\in V$, and let $\mathbf{x} ^\ast : D _n \rightarrow V$ be the corresponding fixed point map. Suppose now that the linear map $\mathcal{D}_{{\bf z}}F({\bf x}_0,{\bf z}_0):\mathbb{R}^n\rightarrow \mathbb{R}^N$ has maximal rank for any ${\bf z} _0\in \left({\bf x}^* \right)^{-1}({\bf x}_0)$, that is, for all the elements ${\bf z}_0 \in D _n$  such that $\mathbf{x} _0= F(\mathbf{x} _0, {\bf z} _0) $. Then:
\begin{enumerate}[(i)]
    \item There exists a neighborhood $V_0$ of ${\bf x}_0$ and an open neighborhood $D_0 ^{i}$ of ${\bf z} _0 ^i $, for each ${\bf z} _0 ^i \in \left({\bf x}^* \right)^{-1}({\bf x}_0)$, such that $F(V_0\times D_0^{i})\subset V_0$ and the restricted state maps  $F_0: V_0\times D_0^{i}\rightarrow V_0$ are all SI-invertible. If the set $\left({\bf x}^* \right)^{-1}({\bf x}_0)$ has finite cardinality, then $V _0 $ can be chosen to be an open set.
    \item Let $I$ be an index set that labels the elements in  $\left({\bf x}^* \right)^{-1}({\bf x}_0)$, that is, $\left({\bf x}^* \right)^{-1}({\bf x}_0)= \left\{{\bf z} _0 ^i \mid i \in I\right\}$. Let $\underline{{\bf z}} _0= \left({\bf z}_t\right)_{t \in \mathbb{Z}_{-}}$ be a sequence such that ${\bf z}_t \in \left({\bf x}^* \right)^{-1}({\bf x}_0)$ for all $t \in \mathbb{Z}_{-}$, that is, ${\bf z}_t = {\bf z} _0^{i _t} $ for some $i _t\in I $ and all $t \in \mathbb{Z}_{-}$. This choice implies that  $U(\underline{{\bf z}} _0) = \underline{{\bf x}} _0$, where $\underline{{\bf x}} _0= (\ldots, \mathbf{x} _0, \mathbf{x} _0)$ is the constant sequence. Let $A_{\underline{{\bf z}} _0}:=\prod_{t \in \mathbb{Z}_{-}}D_0^{i _t}\subset (D_n)^{\mathbb{Z}_{-}}$ be the product set. Then, the restriction of the filter  $U_F|_{A_{\underline{{\bf z}} _0}}:A_{\underline{{\bf z}} _0}\rightarrow (V_0)^{\bb{Z}_-}$ is injective.
\end{enumerate}
\end{proposition}
\begin{proof}
    \textit{(i)} Let us start by studying the situation in which there exists only one element ${\bf z}_0 \in \left({\bf x}^* \right)^{-1}({\bf x}_0)$. By hypothesis, the linear map $\mathcal{D}_{{\bf z}}F({\bf x}_0,{\bf z}_0):\mathbb{R}^n\rightarrow \mathbb{R}^N$ has maximal rank. Hence, the local injectivity theorem guarantees that there exists an open neighborhood $L_{{\bf z}_0}({\bf x}_0)\in \bb{R}^n$ of ${\bf z}_0$ such that the restriction $F_{{\bf x}_0}|_{L_{{\bf z}_0}({\bf x}_0)}:L_{{\bf z}_0}({\bf x}_0)\rightarrow V$ is injective. Since $F$ is differentiable and the rank is a lower semicontinuous function, there exists an open neighborhood $K_{{\bf x}_0}\subset V$ of ${\bf x}_0$ such that  $\mathcal{D}_{{\bf z}}F({\bf x},{\bf z}_0):\mathbb{R}^n\rightarrow \mathbb{R}^N$ has maximal rank for all ${\bf x}\in K_{{\bf x}_0}$, and then, there exist neighborhoods $L_{{\bf z}_0}({\bf x})$ of ${\bf z}_0$ with ${\bf x}\in K_{{\bf x}_0}$ such that the maps $F_{{\bf x}}|_{L_{{\bf z}_0}({\bf x})}:L_{{\bf z}_0}({\bf x})\rightarrow V$ are injective. We show now that there exists an open neighborhood 
    \begin{equation}\label{eq:intersect}
        K_{{\bf z}_0}\subset \bigcap_{{\bf x}\in K_{{\bf x}_0}}L_{{\bf z}_0}({\bf x})
    \end{equation}
    (eventually after redefining $K_{{\bf x}_0}$)
    such that $F_{{\bf x}}|_{K_{{\bf z}_0}}:K_{{\bf z}_0}\rightarrow V$ is injective for all ${\bf x}\in K_{{\bf x}_0}$. \\
    Let $B_r({\bf x}_0)$ be an open ball of radius $r$ centered at ${\bf x}_0$ such that $\overline{B_r({\bf x}_0)}\subset K_{{\bf x}_0}$. Define now the function $f:\overline{B_r({\bf x}_0)}\rightarrow \bb{R}^+$ by $f({\bf x})=\text{inradius}(L_{{\bf z}_0}({\bf x}),{\bf z}_0)$ where $$\text{inradius}(L_{{\bf z}_0}({\bf x}),{\bf z}_0):=\sup \{S>0 \ | \ B_S({\bf z}_0)\subset L_{{\bf z}_0}({\bf x})\}.$$
    Since $L_{{\bf z}_0}({\bf x})$ is open, $\text{inradius}(L_{{\bf z}_0}({\bf x}),{\bf z}_0)$ is strictly positive for each ${\bf x}\in \overline{B_r({\bf x}_0)}$ and hence the set $f(\overline{B_r({\bf x}_0)})$ is bounded below by zero. We show now that even though $f$ is potentially not continuous, it attains a minimum. First, since $f(\overline{B_r({\bf x}_0)})$ is bounded below it has an infimum $a=\inf_{{\bf x}\in \overline{B_r({\bf x}_0)}}\{f({\bf x})\}$. By the approximation property of the infimum, there exists a sequence $\{{\bf x}_n\}\subset \overline{B_r({\bf x}_0)}$ such that $f({\bf x}_n)\underset{n\rightarrow\infty}{\longrightarrow} a$. Since $\overline{B_r({\bf x}_0)}$ is compact, there exists ${\bf l}\in \overline{B_r({\bf x}_0)}$ and a subsequence $\{{\bf x}_{n_{k}}\}$ of $\{{\bf x}_n\}$ such that ${\bf x}_{n_{k}}\underset{k\rightarrow\infty}{\longrightarrow} {\bf l}$. We show now that $f({\bf l}) = a$. Indeed, suppose that it is not true and that $\exists \delta >0$ such that $f({\bf l}) = a+\delta>a$.  But this contradicts the fact that $f({\bf x}_{n_{k}})\underset{k\rightarrow\infty}{\longrightarrow} a$ which implies that $f({\bf l})=a$ and hence $f$ attains a minimum. We emphasize that since $a=f({\bf l})=\text{inradius}(L_{{\bf z}_0}({\bf l}),{\bf z}_0)>0$, additionally, for any ${\bf x}\in \overline{B_r({\bf x}_0)}$, we have 
    $$B_a({\bf z}_0)\subset B_{f({\bf x})}({\bf z}_0)\subseteq L_{{\bf z}_0}({\bf x}), \ \forall {\bf x}\in \overline{B_r({\bf x}_0)},$$
    and hence $B_a({\bf z}_0)\subset \bigcap_{{\bf x}\in \overline{B_r({\bf x}_0)}}L_{{\bf z}_0}({\bf x})$, which implies that $B_a({\bf z}_0)\subset \bigcap_{{\bf x}\in B_r({\bf x}_0)}L_{{\bf z}_0}({\bf x})$ and hence proves \eqref{eq:intersect} by taking $B_a({\bf z}_0)$ as $K_{{\bf z}_0}$ and $B_r({\bf x}_0)$ as $K_{{\bf x}_0}$. Notice that the restriction $F: K_{{\bf x}_0}\times K_{{\bf z}_0}\rightarrow V$ satisfies that $F_{\bf x}:  K_{{\bf z}_0}\rightarrow V$ is injective for all ${\bf x}\in  K_{{\bf x}_0}$ since $K_{{\bf z}_0}\subset L_{{\bf z}_0}({\bf x})$. \\
    We now construct open neighborhoods $V_0 \in K_{{\bf x}_0}$ of ${\bf x}_0$ and $D_0\subset K_{{\bf z}_0}$ of ${\bf z}_0$ that satisfy \textit{(i)} in the statement of the Proposition. Let $0<\epsilon<r(1-c)$ with $r>0$ as in the construction of the open ball $B_r({\bf x}_0)$ above, and $0<c<1$ the contraction constant of $F$ as in Remark \ref{remark:1}. Since $F$ is continuous, there exists $\delta(\epsilon)>0$ such that if ${\bf z}\in B_{\delta(\epsilon)}({\bf z}_0)$, then 
    \begin{equation}
        ||F_{{\bf x}_0}({\bf z})-F_{{\bf x}_0}({\bf z}_0)||\leq \epsilon.
    \end{equation}
    Choose now $\delta <\delta(\epsilon)$ so that $B_{\delta}({\bf z}_0)\subset K_{{\bf z}_0}$. Take now $V_0:=K_{{\bf x}_0}=B_r({\bf x}_0)$ and $D_0:=B_{\delta}({\bf z}_0)$, and we now show that $F(V_0\times D_0)\subset V_0$. Indeed if $({\bf x},{\bf z})\in V_0\times D_0$, then 
    \begin{equation*}
    \begin{split}
        &||F({\bf x},{\bf z})-{\bf x}_0||=||F({\bf x},{\bf z})-F({\bf x}_0,{\bf z}_0)|| \\
        &\leq ||F({\bf x},{\bf z})-F({\bf x}_0,{\bf z})||+||F({\bf x}_0,{\bf z})-F({\bf x}_0,{\bf z}_0)||\\
        &\leq c||{\bf x}-{\bf x}_0||+||F_{{\bf x}_0}({\bf z})-F_{{\bf x}_0}({\bf z}_0)||\\
        &\leq cr+\epsilon \leq r \Rightarrow F({\bf x},{\bf z})\in B_r({\bf x}_0)=K_{{\bf x}_0}, 
    \end{split}
    \end{equation*}
which proves \textit{(i)} when there exists only one element ${\bf z}_0 \in \left({\bf x}^* \right)^{-1}({\bf x}_0)$. 
We now extend the proof that we just wrote down to the case in which the set $\left({\bf x}^* \right)^{-1}({\bf x}_0)$ has arbitrary cardinality and its elements are labeled by a set $I$. We proceed, for each ${\bf z}^i_0\in \left({\bf x}^* \right)^{-1}({\bf x}_0)$, $i \in I $, by mimicking the previous proof to establish the existence of open neighborhoods $V^i_0$ of ${\bf x}_0$ and $D^i_0$ of $ {\bf z}^i_0$ such that $F(V^i_0\times D^i_0) \subset V^i_0$, for all $i\in I$, and additionally, the restrictions $F:V^i_0\times D^i_0\rightarrow V^i_0$ are all SI-invertible.
Redefine $V_0:=\bigcap_{i\in I}V^i_0$. Then the maps $F:V_0\times D_0^{i}\rightarrow V_0$ are also SI-invertible for all $i \in I $. Notice that $V _0$ is a neighborhood of $\mathbf{x} _0$ that is open when $I$ is finite, which proves the claim.

We now proceed to prove part \textit{(ii)} and show that the restriction $U_F|_{A_{\underline{{\bf z}} _0}}:A_{\underline{{\bf z}} _0}\rightarrow (V_0)^{\bb{Z}_-}$ is injective. Indeed, suppose that $\underline{{\bf z} }, \underline{{\bf z} }' \in A_{\underline{{\bf z}} _0} $ are such that $U _F(\underline{{\bf z} })=U _F(\underline{{\bf z} }')= \underline{\mathbf{x}} $ which implies that $F_{\mathbf{x}_{t-1}}({\bf z} _t)=F_{\mathbf{x}_{t-1}}({\bf z}'_t) $, for all $t \in \mathbb{Z}_{-}$. Since in part {\it (i)} we established that all the maps $F:V _0 \times D_0^{i _t} \rightarrow V _0  $, $t \in \mathbb{Z}_{-}$  are SI-invertible, we necessarily have that ${\bf z} _t= {\bf z} _t' $ and hence $\underline{{\bf z} }= \underline{{\bf z}}' $, as required. 
\end{proof}

The proposition that we just proved shows that the filters associated with contractive reservoir maps are locally injective around all the semi-infinite sequences made out of entries in the set $\left({\bf x}^* \right)^{-1}({\bf x}_0)$.  All these sequences are obviously mapped by the filter  $U_F $ to the constant sequence $\underline{\bf x}_0 $. The next proposition shows that these input sequences actually characterize the preimages $U^{-1}_F(\underline{\bf x}_0)$ of constant output sequences  $\underline{\bf x}_0 $ by the filter $U _F$.

\begin{proposition}
\label{characterization constant outputs}
Let $F:V\times D_n \rightarrow V$ be a contractive reservoir map, let $U _F: \left(D_n\right)^{\mathbb{Z}_{-}} \rightarrow V^{\mathbb{Z}_{-}}$ be the corresponding filter, and let $\mathbf{x}^\ast : D_n \rightarrow V $ be the fixed point map.   
\begin{enumerate}[(i)]
\item For any $\mathbf{x}_0 \in V $, we have that
$$U^{-1}_F(\underline{\bf x}_0)= {\rm Seq}\left(\left({\bf x}^* \right)^{-1}({\bf x}_0) \right),$$
where for any subset $S\subset D_n$, ${\rm Seq}(S) \subset D_n ^{\mathbb{Z}_{-}}$ denotes the set of all the sequences in $D_n ^{\mathbb{Z}_{-}}$ that can be constructed by using elements in $S$ (eventually with repetition).
\item If the fixed point map $\mathbf{x}^\ast : D_n \rightarrow V $ is injective then the preimage $U^{-1}_F(\underline{\bf x}_0) $, $\mathbf{x}_0 \in V $, if it is not empty, it contains a unique constant sequence $\underline{\bf z}_0 = (\dots,{\bf z}_0,{\bf z}_0)$, where ${\bf z}_0 \in D_n $ is the unique element such that $\mathbf{x}^\ast ({\bf z} _0)= \mathbf{x} _0 $.
\end{enumerate}
\end{proposition}

\begin{proof}
\textit{(i)} If $\underline{\bf z}\in {\rm Seq}\left(\left({\bf x}^* \right)^{-1}({\bf x}_0) \right) $, then we have that $F(\mathbf{x} _0, {\bf z} _t)= \mathbf{x} _0$, for any $t \in \mathbb{Z}_{-} $, where ${\bf z}_t \in \left({\bf x}^* \right)^{-1}({\bf x}_0)$ are the components of the sequence $\underline{\bf z} $. This implies that the constant sequence $\underline{\bf x}_0$ is a solution of the state equation for the input $\underline{\bf z}$. Since $F$ is contractive, it has the ESP and hence $\underline{\bf x}_0=U _F(\underline{\bf z})$ necessarily, which proves the inclusion ${\rm Seq}\left(\left({\bf x}^* \right)^{-1}({\bf x}_0) \right) \subset U^{-1}_F(\underline{\bf x}_0)$. Conversely, let $\underline{\bf z} \in U^{-1}_F(\underline{\bf x}_0)$. This implies that $F(\mathbf{x} _0, {\bf z} _t)= \mathbf{x} _0$, for any $t \in \mathbb{Z}_{-} $, which by definition guarantees that ${\bf z}_t \in \left({\bf x}^* \right)^{-1}({\bf x}_0) $ for any $t \in \mathbb{Z}_{-}$ and hence that  $\underline{\bf z} \in {\rm Seq}\left(\left({\bf x}^* \right)^{-1}({\bf x}_0) \right) $. \textit{(ii)} follows from the fact that when $\mathbf{x}^\ast   $ is injective then the set $\left({\bf x}^* \right)^{-1}({\bf x}_0) $ is a singleton.
\end{proof}

\begin{remark}\normalfont
Proposition \ref{characterization constant outputs} connects with a new extension of the ESP that has been recently proposed, called nonstationary ESP \cite{kobayashi2024extending,kobayashi2024coherence}. This tool is introduced to determine whether a driven dynamical system with a time-varying accessible state space (i.e. nonstationary) can be useful for time-series processing, with a particular focus on QRC models. In fact, a connection can be established between the nonstationary ESP and our results. Following the proof of Theorem III.4 in Ref.~\cite{kobayashi2024extending}, it can be shown that having  
$\vertiii{p({\bf z})}<1- \epsilon$, for all ${\bf z} \in D_n$, and some $\epsilon >0 $,  and injectivity in the fixed-point function [Proposition \ref{characterization constant outputs} \textit{(ii)}] is sufficient to guarantee the nonstationary ESP.
\end{remark}

\begin{remark}
\normalfont
We conclude this discussion about constant inputs and outputs by emphasizing that Proposition \ref{characterization constant outputs} shows that it is not only constant inputs that may be mapped to constant outputs. Additionally, the fixed-point functions of contractive reservoir systems are, in general, noninjective, which leads to several constant input sequences producing the same constant output sequence. However, even in that situation, Proposition \ref{local proposition 1} shows that the reservoir filter could still be locally injective in a neighborhood of each of these input sequences.  Example \ref{ex:period} below is a case where periodicity in the input encoding breaks global filter injectivity, and despite having a noninjective fixed point map, local injectivity can still be shown in a neighborhood of each of the sequences in ${\rm Seq}\left(\left({\bf x}^* \right)^{-1}({\bf x}_0) \right) $. 
\end{remark}

The proposition below formulates the statement in Proposition \ref{local proposition 1} for the particular case of SAS reservoir maps.

\begin{proposition}
\label{prop:2} 
Let $F:V\times D_n \rightarrow V$ be a differentiable SAS system given by $F(\mathbf{x}, {\bf z})=p({\bf z}){\bf x}+q({\bf z})$ and let us assume that $\vertiii{p({\bf z})}<1- \epsilon$, for all ${\bf z} \in D_n$, and some $\epsilon >0 $, which guarantees that $F$ has the ESP. Let $\mathbf{x}_0 \in V$ and assume that for any ${\bf z}_0\in \left(\mathbf{x} ^\ast\right)^{-1}(\mathbf{x} _0) $ the linear maps 
\begin{equation}
\label{rank condition local}
\left(\mathcal{D}p({\bf z}_0)(\cdot)\right){\bf x}_0+\mathcal{D}q({\bf z}_0):\mathbb{R}^n\rightarrow \mathbb{R}^N
\end{equation}
have all rank $n$ and hence are injective. Then, there exist a neighborhood $V_0$ of ${\bf x}_0$ and open neighborhoods $D_0^{i}$ around all points  $\mathbf{z} _0^{i} \in \left(\mathbf{x} ^\ast\right)^{-1}(\mathbf{x} _0) \subset D _n $, $i \in I $, ($I$ is and index set for the elements in $\left(\mathbf{x} ^\ast\right)^{-1}(\mathbf{x} _0) $) with the following property:  let $\underline{{\bf z}} _0= \left({\bf z}_t\right)_{t \in \mathbb{Z}_{-}}$ be a sequence such that ${\bf z}_t \in \left({\bf x}^* \right)^{-1}({\bf x}_0)$ for all $t \in \mathbb{Z}_{-}$, that is, ${\bf z}_t = {\bf z} _0^{i _t} $ for some $i _t\in I $ and all $t \in \mathbb{Z}_{-}$. Let $A_{\underline{{\bf z}} _0}:=\prod_{t \in \mathbb{Z}_{-}}D_0^{i _t}\subset (D_n)^{\mathbb{Z}_{-}}$ be the product set. Then, the restriction of the filter  $U_F|_{A_{\underline{{\bf z}} _0}}:A_{\underline{{\bf z}} _0}\rightarrow (V_0)^{\bb{Z}_-}$ is injective.
\end{proposition}

SAS systems are an example in which the fixed point map $\mathbf{x}^\ast  $ can be explicitly written down. Indeed, in that case, if $p$ is contractive, then the corresponding filter is given by \eqref{eq:filter_x}. The image by that filter of the constant sequence $\underline{\bf z}=(\dots,{\bf z},{\bf z})$ implies that 
\begin{equation}\label{eq:x_0}
\mathbf{x}^\ast({\bf z}) = \sum^\infty_{j=0}\left(\prod^{j-1}_{k=0}p({\bf z})\right)q({\bf z})=(I-p({\bf z}))^{-1}q({\bf z}). 
\end{equation}
If $p$ is contractive and hence expression \eqref{eq:x_0} holds true, then Proposition \ref{prop:2} yields a local injectivity result on the neighborhoods of constant input sequences whose hypotheses are particularly simple to verify.
\begin{corollary}\label{cor:1}
In the setup of Proposition \ref{prop:2}, let ${\bf z}\in D_n$ and $\underline{{\bf z}}=(\dots,{\bf z},{\bf z})$ be the corresponding constant sequence. Suppose that the linear map $$\left(\mathcal{D}p({\bf z})(\cdot)\right){\bf x} ^\ast ({\bf z})+\mathcal{D}q({\bf z}):\mathbb{R}^n\rightarrow \mathbb{R}^N$$ has rank $n$. Then, there is an open neighborhood $A_{\underline{{\bf z}}}$ of $\underline{{\bf z}}$ in $(D_n)^{\mathbb{Z}_-}$ such that $$U_F|_{A_{\underline{{\bf z}}}}:A_{\underline{{\bf z}}} \rightarrow (V)^{\mathbb{Z}_-}$$ is injective.
\end{corollary}

\begin{example}
\normalfont
Let us now illustrate Corollary \ref{cor:1}. The same will be done with the content of Proposition \ref{prop:2} once the family of contracted-encoding channels is introduced later on in Sec. \ref{sec:c-e quatum channels}. 
 Consider the system introduced in Eq.~(62) in Sec. IV of Ref.~\cite{martinez2023quantum} for one-dimensional inputs. The Markovian master equation that governs the dynamics is:
\begin{equation}
\dot{\rho}=-i[H(z_t),\rho]+\gamma L\rho L^{\dagger}-\frac{\gamma}{2}\{L^{\dagger}L,\rho\}, 
\end{equation}
where $H(z_t)=z_t\sigma^x/2$ is the input-dependent Hamiltonian, $L=\sigma^-$ is the jump operator, and $\gamma$ is the decay rate. For each time step, the input $z_t$ is kept constant, and the system is evolved for a time $\Delta \tau$. Integrating the master equation over this time interval, the quantum map can be written as
\begin{equation}
\rho_t = e^{\mathcal{L}(z_t)\Delta \tau}\rho_{t-1},
\end{equation}
where $\mathcal{L}(z_t)$ is the Lindbladian of the master equation. The procedure to obtain the matrix expression $\widehat{T}$ of the quantum channel can be found in Ref.~\cite{martinez2023quantum}. The outcome is
\begin{equation}\label{eq:T_good}
\left( \begin{matrix} 1   \\ 
\braket{\sigma^x}_t \\
\braket{\sigma^y}_t \\
\braket{\sigma^z}_t
\end{matrix} \right) =\left( \begin{matrix} 1 & 0 & 0 & 0  \\ 

0 & \widehat{T}_{22} & 0 & 0 \\
\widehat{T}_{31} & 0 &  \widehat{T}_{33} & \widehat{T}_{34} \\
\widehat{T}_{41} &  0 &  \widehat{T}_{43} & \widehat{T}_{44}
\end{matrix} \right) \left( \begin{matrix} 1   \\ 
\braket{\sigma^x}_{t-1} \\
\braket{\sigma^y}_{t-1} \\
\braket{\sigma^z}_{t-1}
\end{matrix} \right),
\end{equation}
see Eq.~(63) in Ref.~\cite{martinez2023quantum} for the detailed form of the matrix elements. This model is a promising candidate for being a valuable QRC system since it has an input-dependent fixed point, which in SAS representation, is 
\begin{equation}\label{eq:x_master}
\textbf{x}^*(z) =\frac{1}{\sqrt{2}}\left( \begin{matrix} \tr\left(\sigma^x\rho^*(z)\right)   \\ 
\tr\left(\sigma^y\rho^*(z)\right) \\
\tr\left(\sigma^z\rho^*(z)\right)
\end{matrix} \right)=\frac{1}{\sqrt{2}} \left( \begin{matrix} 0   \\ 
-\frac{2z\gamma}{2z^2+\gamma^2} \\
-\frac{\gamma^2}{2z^2+\gamma^2}
\end{matrix} \right).
\end{equation}
We show in Fig.~\ref{Fig:1} that the rank condition of Corollary \ref{cor:1} is met for almost all values of $\gamma$ and $z$. Here, the rank condition is studied for a unidimensional input, so the maximum rank can be only $n=1$, what is translated in requiring a nonzero vector norm. Notice also that \eqref{eq:x_master} is injective when $\gamma^2 \neq 16 z^2$.
\begin{figure}[h]
\captionsetup[subfigure]{}
\begin{center}
\includegraphics[scale=0.8]{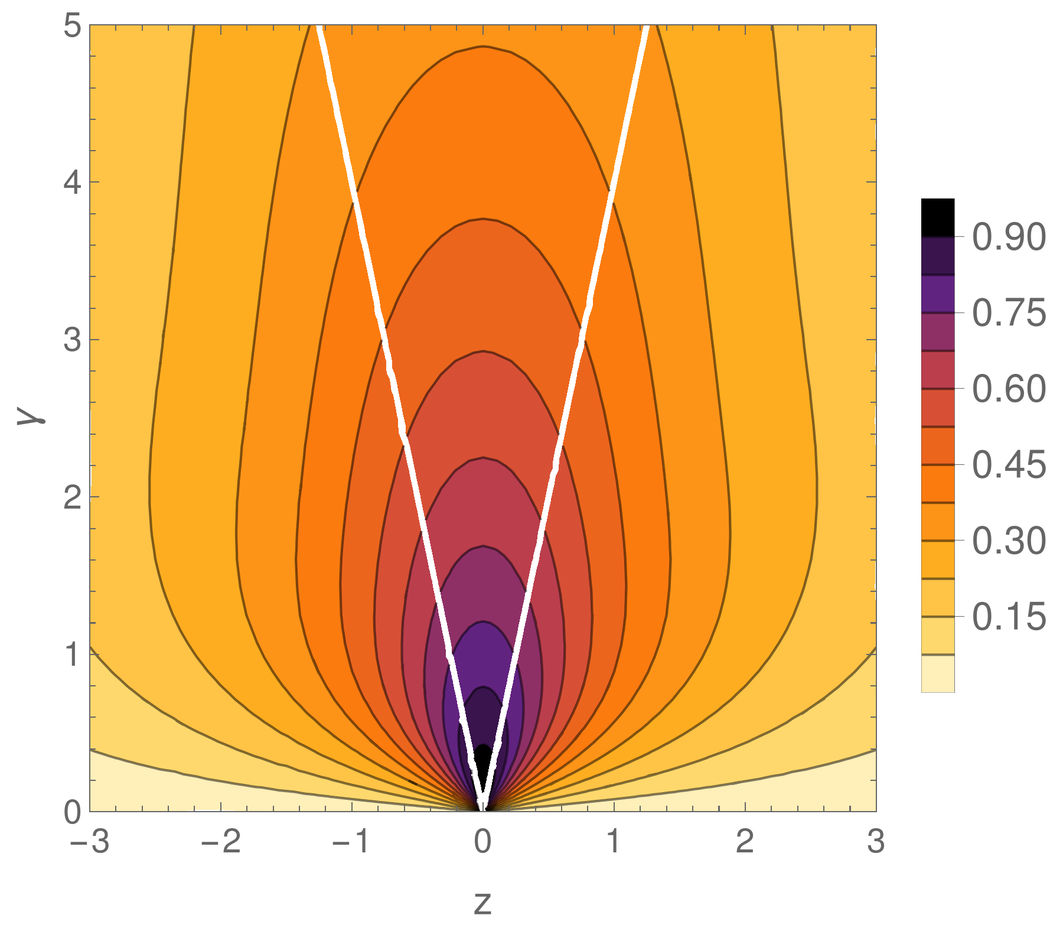}
\caption{Vector norm for the rank condition with $\Delta \tau = 1$. The white lines correspond to the singularity $\gamma^2 = 16 z^2$.}\label{Fig:1}
\end{center}
\end{figure}
\end{example}

\begin{remark}
\normalfont
The functional form of the input encoding is crucial in guaranteeing the rank condition in \eqref{rank condition}. For instance, just introducing a quadratic input encoding in the previous example as $H(z_t)=z^2_t\sigma^x/2$, reduces the rank in Corollary \ref{cor:1} to zero at $z=0$ regardless the values of $\gamma$ and $\Delta \tau$.
\end{remark}

\subsection{Contracted-encoding quantum channels}
\label{sec:c-e quatum channels}

We now focus our input-dependence analysis on an important family of QRC models that we call {\it contracted-encoding quantum channels}. Let $T: \mathcal{S}(\mathcal{H}) \times D_n\rightarrow \mathcal{S}(\mathcal{H})$ be the state-space transformation defined by
\begin{equation}\label{eq:rho}
\rho_t=T(\rho _t,{\bf z} _t):=\mathcal{E}(\mathcal{J}(\rho_{t-1},{\bf z}_t)),
\end{equation}
where $\mathcal{E}: \mathcal{S}(\mathcal{H})\rightarrow \mathcal{S}(\mathcal{H})$ is a strictly contractive map and $\mathcal{J}:\mathcal{S}(\mathcal{H}) \times D_n\rightarrow \mathcal{S}(\mathcal{H})$ is a CPTP map that encodes the input information. Note that choosing the contraction constant of $\mathcal{E}$ small enough is a way to typically ensure, using Proposition 3 in Ref.~\cite{martinez2023quantum},  that $T$ has the ESP. This rather broad family is very present in the QRC literature: as a composition of a dissipative channel with an amplitude encoding map \cite{fujii2017harnessing,chen2019learning,gies2024exploring,kora2024frequency,palacios2024role,mujal2023time,franceschetto2024harnessing}; as a quantum circuit with noise \cite{suzuki2022natural,kubota2023temporal,monzani2024leveraging}, mid-circuit measurements \cite{yasuda2023quantum,fuchs2024quantum,hu2024overcoming} or a reset-rate channel \cite{chen2020temporal,chen2021nonlinear,molteni2023optimization,mifune2024effects}.  Analytical results have also been derived for the particular case of $\mathcal{J}$ being unitary \cite{kobayashi2024coherence}.

The structure \eqref{eq:rho} of contracted-encoding channels allows us to be more specific on the criteria that we have spelled out so far regarding input dependence. 

\medskip

\subsubsection{\bf Input-independent contracted-encoding channels} 
We start with the criterion about constant filters that we introduced in Theorem \ref{th:1} and that characterizes input-independent architectures that should be avoided.

\begin{theorem} 
\label{th:2}
Let $T: \mathcal{B}(\mathcal{H}) \times D_n\rightarrow \mathcal{B}(\mathcal{H})$ be a contracted-encoding system as in~\eqref{eq:rho} and  suppose that the input space $D_n$ is compact.   Let $\rho^*_{\mathcal{E}} \in  \mathcal{S}(\mathcal{H})$ be the unique fixed point of the strictly contractive quantum channel  $\mathcal{E}$. Then:
\begin{enumerate}[(i)]
\item  If $\rho^*_{\mathcal{E}} $ is also a fixed point of the CPTP map $\mathcal{J}$,that is, $\rho^*_{\mathcal{E}}=\mathcal{J}(\rho^*_{\mathcal{E}},{\bf v})$ for all ${\bf v}\in D_n$, then the corresponding filter $U_T $ is constant. More specifically $U_T({\bf z}) _t=\rho^*_{\mathcal{E}}$ for all  ${\bf z}\in (D_n)^{\mathbb{Z}}$ and $t \in \Bbb Z $. 
\item  Assume the filter $U_T $ is constant, that is $U_T({\bf z}) _t=\rho^*_T$ for all  ${\bf z}\in (D_n)^{\mathbb{Z}}$, $t \in \Bbb Z $. Assume, additionally, that $\mathcal{E}$ has an inverse map $\mathcal{E}^{-1}: \mathcal{B}(\mathcal{H})\rightarrow \mathcal{B}(\mathcal{H})$. The following conditions are equivalent:
\begin{enumerate}[1.]
\item $\rho^*_{\mathcal{E}}=\rho^*_T$,
\item $\rho':=\mathcal{J}(\rho^*_T,{\bf v})=\rho^*_T$, and hence $\rho'$ does not depend on ${\bf v}\in D_n$.
\item $\rho^*_{\mathcal{E}}=\rho'$,
\end{enumerate}
in which case, $\rho^*_{\mathcal{E}}=\rho'=\rho^*_T$. If, additionally, $\mathcal{J}$ is a unitary map such that $\mathcal{J}(\rho,{\bf v})=\mathcal{U}({\bf v})\rho\mathcal{U}^{\dagger}({\bf v})$, then these statements are also equivalent to 
\begin{enumerate}[4.]
\item $[\rho',\mathcal{U}({\bf v})]=[\rho^*_T,\mathcal{U}({\bf v})]=0$.
\end{enumerate}
\end{enumerate}
\end{theorem} 

\begin{proof}  
\textit{ (i)} Following Theorem \ref{th:1}, the filter $U _T $ is constant if and only if the contractive maps $T(\cdot , {\bf z})$ have a unique and input-independent fixed point. Then, if $\rho^*_{\mathcal{E}}$ is the only fixed point of $\mathcal{E}$ and one of the fixed points of $\mathcal{J}$, it becomes the only fixed point of the total map and the filter is hence necessarily constant. 

\noindent \textit{ (ii)} If the filter $U_T$ is constant then, by Theorem \ref{th:1}, $\rho^*_T$ satisfies that $T(\rho^*_T,{\bf v})=\rho^*_T$, for all ${\bf v} \in D_n$.
We now show that the image of the fixed point $\rho^*_T$ by the maps $\mathcal{J}(\cdot , {\bf v})$ yields an input-independent density matrix, and we will subsequently show the equivalence between the different statements. We start by assuming that $\mathcal{J}(\rho^*_T,{\bf v})=\rho'({\bf v})$, where  $\rho'({\bf v})$ is in principle input dependent. Since $\rho^*_T$ is the fixed point of $T$, the following equation must hold true for all ${\bf s}, {\bf v}\in D_n$:
\begin{equation}
\mathcal{E}(\rho'({\bf s}))=\rho^*_T=\mathcal{E}(\rho'({\bf v})).
\end{equation}
Using the existence of the inverse map of $ \mathcal{E}$, we have that $\rho'({\bf s})=\rho'({\bf v})$ for all ${\bf z}, {\bf v}\in D_n$. Therefore, $\rho'$ must be input independent. We now prove the equivalence between the statements. 

\noindent $\textit{1.} \leftrightarrow \textit{2.}$  If $\rho^*_{\mathcal{E}}=\rho^*_T$, then we have that $\rho^*_T=\mathcal{E}(\rho^*)$ is the only fixed point of  $\mathcal{E}$. We also have that $\rho^*_T=\mathcal{E}(\rho')$, so $\mathcal{E}(\rho^*_T)=\mathcal{E}(\rho')$. Assuming the existence of the inverse of $\mathcal{E}$, the equality $\rho'=\rho^*_T$ holds, and therefore $\rho'=\rho^*_T=\rho^*_{\mathcal{E}}$. If we start by assuming that $\rho'=\rho^*_T$, then $\rho^*_T=\mathcal{E}(\rho')=\mathcal{E}(\rho^*_T)$. Since the map $\mathcal{E}$ only has one fixed point, we conclude that $\rho^*_{\mathcal{E}}=\rho^*_T$, and again $\rho^*_{\mathcal{E}}=\rho^*_T=\rho'$. 

\medskip

\noindent $\textit{2.}\leftrightarrow \textit{3.}$ If $\rho'=\rho^*_T$, then, as we just saw, $\rho^*_{\mathcal{E}}=\rho^*_T=\rho'$. If $\rho^*_{\mathcal{E}}=\rho'$, then $\mathcal{E}(\rho')=\rho'$, so $\rho^*_T=\mathcal{E}(\mathcal{J}(\rho^*_T,{\bf v}))=\mathcal{E}(\rho')=\rho'$. 

\medskip

\noindent $\textit{2.}\leftrightarrow \textit{4.}$ First, notice that $[\rho',\mathcal{U}({\bf v})]=\mathcal{U}({\bf v})[\rho^*_T,\mathcal{U}({\bf  v})]\mathcal{U}^\dagger({\bf v})$, so $[\rho',\mathcal{U}({\bf v})]=0$ if and only if $[\rho^*_T,\mathcal{U}({\bf v})]=0$. If $[\rho',\mathcal{U}({\bf v})]=[\rho^*_T,\mathcal{U}({\bf v})]=0$, then we see that $\rho'=\mathcal{U}({\bf v})\rho^*_T\mathcal{U}^\dagger({\bf v})=\mathcal{U}({\bf v})\mathcal{U}^\dagger({\bf v}) \rho^*_T=\rho^*_T$. Conversely, if $\rho'=\rho^*_T$, then $\rho^*_T=\mathcal{U}({\bf v})\rho^*_T \mathcal{U}^\dagger({\bf v})\rightarrow \rho^*_T \mathcal{U}({\bf v})=\mathcal{U}({\bf v})\rho^*_T\rightarrow [\rho^*_T,\mathcal{U}({\bf v})]=0$.
\end{proof}

\begin{remark}\normalfont
Even though the hypothesis on the existence of the inverse of $\mathcal{E}$ may seem restrictive, it is actually met by several commonly used strictly contractive quantum channels, like the depolarizing and the amplitude-damping channels. We emphasize that the inverse of a quantum channel does not need to be a quantum channel. It will be the case when it describes a unitary evolution (see, for instance, Theorem 3.4.1 in Ref.~\cite{rivas2012open}).
\end{remark}

Theorem \ref{th:2} \textit{ (i)} describes a sufficient condition for having a constant filter that does not require the existence of an inverse for $\mathcal{E}$. In that case we can also deduce that $\rho'=\mathcal{J}(\rho^*_T,{\bf v})=\mathcal{J}(\rho^*_{\mathcal{E}},{\bf v})=\rho^*_{\mathcal{E}}=\rho^*_T$, for all $\mathbf{v} \in D _n$. However, we emphasize that it is not necessary to have $\rho^*_{\mathcal{E}}=\mathcal{J}(\rho^*_{\mathcal{E}},{\bf v})$ to obtain a constant filter, as we will see in the examples below. Part \textit{(ii)} addresses the reverse implication by describing the underlying connections between the CPTP map $\mathcal{J}$ and the dissipative map $\mathcal{E}$ when we know that the filter is constant and, additionally, any of the statements is fulfilled. In that particular case, we have shown that $\rho'=\rho^*_\mathcal{E}=\rho^*_T$. However, the following statement is false: if the filter $U_T=\rho^*_T$ is constant, that is, $U _T(\underline{{\bf z}})_t=\rho^*_T$ for all $\underline{{\bf z}}\in (D_n)^{\mathbb{Z}}$ and $t \in \Bbb Z$, then $\rho'=\rho^*_\mathcal{E}=\rho^*_T$.  Counterexamples can be constructed where the filter is constant and $\mathcal{E}$ has an inverse, but $\rho'\neq \rho^*_\mathcal{E}\neq \rho^*_T$.

\begin{example} \label{ex:AD}
\normalfont
Let us consider a quantum channel defined using unitary maps $\mathcal{U}({\bf z}_t)$ in the construction of $\mathcal{J}$:
\begin{equation}\label{eq:rho_u}
\rho_t=T(\rho _{t-1},{\bf z} _t)=\mathcal{E}(\mathcal{U}({\bf z}_t)\rho_{t-1}\mathcal{U}^{\dagger}({\bf z}_t)).
\end{equation}
We now consider the case of a single qubit and a unidimensional input series. Let the unitary map be defined by
\begin{equation*}
\mathcal{U}(z_t)=e^{-i\frac{z_t}{2}\sigma_z}H,
\end{equation*} where $H$ is the Hadamard gate, $\sigma_z$ is the Pauli matrix in the $z$ direction and $e^{-i\frac{z_t}{2}\sigma_z}$ is a rotation around the $z$ axis with angle $z_t$. Consider the strictly contractive channel 
\begin{equation*}
\mathcal{E}(\rho)=e^{-i\frac{\theta}{2}\sigma_y}\mathcal{E}_{\text{AD}}(H\rho H)e^{i\frac{\theta}{2}\sigma_y},
\end{equation*}
where $\theta\in (0,\pi/2]$, is the angle of the rotation gate in the $y$ axis and $\mathcal{E}_{\text{AD}}$ is the amplitude damping channel, 
\begin{equation}
\mathcal{E}_{\text{AD}}(\rho) = 
\begin{pmatrix} (1-\lambda)\rho_{00} &\sqrt{1-\lambda} \rho_{01}\\
\sqrt{1-\lambda} \rho_{10} & \lambda\rho_{00}+\rho_{11}
\end{pmatrix},
\end{equation}
with $0\leq \lambda \leq 1$. Let us fix its rate to $\lambda = 1-\cos^2(\theta)$. The inverse map $\mathcal{E}^{-1}$ is defined as
\begin{equation*}
\mathcal{E}^{-1}(\rho)=H\mathcal{E}^{-1}_{\text{AD}}(e^{i\frac{\theta}{2}\sigma_y}\rho e^{-i\frac{\theta}{2}\sigma_y})H,
\end{equation*}
where $\mathcal{E}^{-1}_{\text{AD}}$ would be a quantum channel if and only if $\lambda=0$ (i.e. when it describes a unitary evolution). Under all the previous conditions, one can show that the filter $U_T$ is constant, with density matrices
\begin{equation*}\begin{split}
& \rho^*_T= \frac{1}{2}\begin{pmatrix} 1 & \sin(\theta)\\
\sin(\theta) & 1
\end{pmatrix}, \quad \rho'= \frac{1}{2}\begin{pmatrix} 1 +\sin(\theta)&0\\
0 & 1-\sin(\theta)
\end{pmatrix},\\ 
&\rho^*_{\mathcal{E}}= \frac{1}{2}\begin{pmatrix} 1+\frac{1+2\cos(\theta)+\cos(2\theta)}{f(\theta)} &1+\frac{2\sin(\theta)-2}{f(\theta)}\\
1+\frac{2\sin(\theta)-2}{f(\theta)} & 1-\frac{1+2\cos(\theta)+\cos(2\theta)}{f(\theta)}
\end{pmatrix},
\end{split}
\end{equation*}
where $f(\theta) =3+2\cos(\theta)+\cos(2\theta)+\sin(2\theta)$, $\rho^*_T$, $\rho^*_{\mathcal{E}}$ and $\rho'$ are input independent.
\end{example}

\begin{example}\label{ex:reset}
\normalfont
The {\it reset-rate quantum reservoir map} is
\begin{equation}\label{eq:rho_reset}
\rho_t=\mathcal{E}_{\epsilon}(\mathcal{J}(\rho _t,{\bf z} _t))=(1-\epsilon)\mathcal{J}(\rho_{t-1},{\bf z}_t)+\epsilon\sigma,
\end{equation}
where $\mathcal{J}(\rho_{t-1},{\bf z}_t)$ is a CPTP map, $\sigma$ an arbitrary density matrix, and $0\leq \epsilon<1$. It defines a unique filter 
\begin{equation}\label{eq:filter_reset}
U(\underline{{\bf z}})_t=\epsilon\sum^{\infty}_{j=0}(1-\epsilon)^j\overleftarrow{\prod}^{0}_{k=j-1} \mathcal{J}(\sigma,{\bf z}_{t-k}),
\end{equation}
where 
\begin{equation}
\overleftarrow{\prod}^{0}_{k=j-1}\mathcal{J}(\sigma,{\bf z}_{t-k}):=\mathcal{J}(\cdot,{\bf z}_{t})\circ \cdots \circ \mathcal{J}(\sigma,{\bf z}_{t-j+1})
\end{equation} and
\begin{equation}\label{eq:j=0}
\overleftarrow{\prod}^{0}_{k=-1} \mathcal{J}(\sigma,{\bf z}_{t-k}):=\sigma.
\end{equation}
The fixed point of $\mathcal{E}_{\epsilon}$ is $\sigma$, i.e. $\mathcal{E}_{\epsilon}(\sigma)=\sigma$, and there exist an inverse map $\mathcal{E}_{\epsilon}^{-1}: \mathcal{B}(\mathcal{H})\rightarrow \mathcal{B}(\mathcal{H})$ for $0\leq \epsilon < 1$ such that $\mathcal{E}_{\epsilon}^{-1}\circ \mathcal{E}_{\epsilon}=\mathcal{E}_{\epsilon}\circ \mathcal{E}_{\epsilon}^{-1}=\mathcal{I}$, given by 
\begin{equation}
\mathcal{E}_{\epsilon}^{-1}(\rho) = \frac{1}{1-\epsilon}\left(\rho-\epsilon \sigma\right).
\end{equation}
Notice that this map is not positive semidefinite in general. Take for example a single qubit, with $\sigma = \ket{0}\bra{0}$ and input state $\rho=\ket{1}\bra{1}$. The eigenvalues of the matrix $\mathcal{E}_{\epsilon}^{-1}(\rho)$ are $\lambda_1=1/(1-\epsilon)$ and $\lambda_2=-\epsilon/(1-\epsilon)$, so $\mathcal{E}_{\epsilon}^{-1}(\rho)$ is not a positive-semidefinite matrix.

We now define a specific map $\mathcal{J}$ for a single qubit with a unidimensional input $v\in D_1$:
\begin{equation}
\mathcal{J}(\rho,v)=\mathcal{E}_{\text{d}}(\mathcal{U}(v)\rho\mathcal{U}^{\dagger}(v)),
\end{equation}
where 
\begin{equation}
\mathcal{E}_{\text{d}}(\rho) = 
\begin{pmatrix} \rho_{00} & \lambda \rho_{01}\\
\lambda \rho_{10} & \rho_{11}
\end{pmatrix}
\end{equation}
is the dephasing channel ($0\leq \lambda \leq 1$) and 
\begin{equation}
\mathcal{U}(v) = e^{-i\frac{v}{2}\sigma_x} 
\end{equation}
is a rotation on the $x$ axis. Notice that the reservoir map could be also written as $T(\rho,v)=\tilde{\mathcal{E}}(\mathcal{U}(v)\rho\mathcal{U}^{\dagger}(v))$, where $\tilde{\mathcal{E}} = \mathcal{E}_{\epsilon}\circ \mathcal{E}_{\text{d}}$.
We now set the parameters $\epsilon=\lambda=1/2$ and $\sigma = \begin{pmatrix} 1/2 & 1/2\\
1/2 & 1/2
\end{pmatrix}$. Under these very specific conditions, one can show that the filter $U_T$ is constant, with density matrices
\begin{equation*}\begin{split}
& \rho^*= \begin{pmatrix} 1/2 &  1/3\\
1/3 &   1/2
\end{pmatrix}, \quad \rho'= \begin{pmatrix}   1/2& 1/6\\
1/6 &  1/2
\end{pmatrix}, \\ 
&\rho^*_{\mathcal{E}_\epsilon}= \sigma =\frac{1}{2}\begin{pmatrix}  1& 1\\
1 &  1
\end{pmatrix},
\end{split}
\end{equation*}
where $\rho^*$, $\rho^*_{\mathcal{E}_\epsilon}$, and $\rho'$ are obviously input independent.
\end{example}

\subsubsection{\bf Filter injectivity with contracted-encoding channels}
We now focus on the filter invertibility of the contracted-encoding channels introduced in~\eqref{eq:rho}. Note first that the SAS representation of these channels has the form 
\begin{equation}
\textbf{x}_t = p(\textbf{z}_t)\textbf{x}_{t-1}+q(\textbf{z}_t) = p_\mathcal{E}p_\mathcal{J}(\textbf{z}_t) \textbf{x}_{t-1}+p_\mathcal{E}q_\mathcal{J}(\textbf{z}_t) + q_\mathcal{E},
\end{equation}
where $p_\mathcal{E}, p_\mathcal{J}(\textbf{z}_t): \mathbb{R}^N \rightarrow \mathbb{R}^N $ are linear and $q_\mathcal{E}, q_\mathcal{J}(\textbf{z}_t)\in  \mathbb{R}^N$.
Proposition \ref{prop:1} guarantees the invertibility of the corresponding filters whenever the linear maps
\begin{equation}
\label{maps for rank}
p_\mathcal{E}\Bigl(\bigl(\mathcal{D}p_\mathcal{J}(\textbf{z})(\cdot)\bigr) \textbf{x}+\mathcal{D}q_\mathcal{J}(\textbf{z}) \Bigr):\mathbb{R}^n\rightarrow\mathbb{R}^N,
\end{equation}
have rank $n$ for all {${\bf z} \in D_n$ and $\mathbf{x} \in V_R$}. 

There are two particular cases in which the filter invertibility condition given by \eqref{maps for rank} can be simplified. First, if the strictly contractive channel $\mathcal{E}$ is invertible, then it suffices to impose the maximal rank condition on the maps $\mathcal{D}p_\mathcal{J}(\textbf{z})(\cdot) \textbf{x}+\mathcal{D}q_\mathcal{J}(\textbf{z}):\mathbb{R}^n\rightarrow\mathbb{R}^N$. 

Second, it is also worth considering the case in which the contracted-encoding channel has a {unital (or unitary)} encoding map $\mathcal{J}$.
In this case, we find that $q_\mathcal{J}(\textbf{z}) = 0$, and hence the rank condition is further simplified and concerns the linear maps:
\begin{equation}\label{eq:rankJunitary}
p_\mathcal{E}\left(\mathcal{D}p_\mathcal{J}(\textbf{z})(\cdot) \textbf{x}\right):\mathbb{R}^n\rightarrow \mathbb{R}^N.
\end{equation}
As before, an invertible $\mathcal{E}$ would reduce the problem to evaluating the rank of the maps $\left(\mathcal{D}p_\mathcal{J}(\textbf{z})(\cdot)\right) \textbf{x}:\mathbb{R}^n\rightarrow \mathbb{R}^N$. Notice that $\textbf{x}=\textbf{0}$ must be avoided as, in that case, the rank is trivially zero.

If we focus exclusively on the local filter invertibility for constant output sequences, Proposition \ref{prop:2} and Corollary \ref{cor:1} apply. In particular, for Corollary \ref{cor:1} we can further simplify the rank condition by requiring it on the maps
\begin{equation}
p_\mathcal{E}\left(\mathcal{D}p_\mathcal{J}(\textbf{z})(\cdot)\right) (I-p_\mathcal{E}p_\mathcal{J}(\textbf{z}))^{-1}q_\mathcal{E}:\mathbb{R}^n\rightarrow \mathbb{R}^N.
\end{equation}

\begin{example}
\normalfont
We illustrate the content of Proposition \ref{prop:1} in the contracted-encoding case with a simple QRC model with a unital channel $\mathcal{J}$ that uses again the reset-rate channel as strictly contractive map $\mathcal{E}$. To compute its matrix representation, we must extend it to be trace-preserving for all operators in $\mathcal{B}(\mathcal{H})$. Let us define the CPTP map $\mathcal{E}': \mathcal{B}({\mathcal H}) \longrightarrow \mathcal{B}({\mathcal H}) $
\begin{equation} \label{eq:dep_B}
\mathcal{E}'(A)=(1-\epsilon)A+\epsilon\tr(A)\sigma,
\end{equation}
where we will study the case of a single qubit, and we fix $\sigma = \ket{0}\bra{0}$. The SAS representation associated to this map is  
\begin{equation}\label{eq:SASepsilon}
q_{\mathcal{E}'}=\begin{pmatrix}  0 \\
0 \\
\epsilon
\end{pmatrix},\quad p_{\mathcal{E}'}=\begin{pmatrix}  1-\epsilon& 0 & 0\\
0 &  1-\epsilon & 0 \\
0 &  0 & 1-\epsilon    
\end{pmatrix}.
\end{equation} 
We  take the depolarizing channel with unidimensional input $z$ as the unital channel $\mathcal{J}: \mathcal{B}({\mathcal H}) \longrightarrow \mathcal{B}({\mathcal H})$, that is
\begin{equation}
\mathcal{J}(A,z)=zA+(1-z)\tr(A)\frac{I}{d},
\end{equation}
which is again a reset-rate channel with $\sigma = I/d$ for $z \in [0,1]$. The SAS representation of this quantum channel is the following:
\begin{equation}
q_{\mathcal{J}}(z)=\textbf{0},\quad p_{\mathcal{J}}(z)=\begin{pmatrix}  z& 0 & 0\\
0 &  z & 0 \\
0 &  0 & z    
\end{pmatrix}.
\end{equation}
The linear map~\eqref{eq:rankJunitary} is, in this case, given by
\begin{equation}
p_{\mathcal{E}'}\left(\mathcal{D}p_\mathcal{J}({z})(s) \textbf{x}\right) 
=  (1-\epsilon){\bf x}s,
\end{equation}
for any $s\in \mathbb{R}$ and any $\mathbf{x}=(x _1, x _2, x _3)^{\top}  \in V_R$. The rank of this linear map is zero if and only if ${\bf x}=0$. The expression \eqref{eq:filter_x} for the SAS filter takes the form
\begin{equation*}
U(\underline{{\bf z}})_t= {\epsilon}\sum_{j=0}^{\infty} (1- \epsilon)^j\prod^{j-1}_{k=0}z_{t-k}  \left(0,0,1\right)^{\top}, \ t \in \Bbb Z_-.
\end{equation*}
The series contains only positive terms (even when we use the zero sequence, where $U(\underline{{\bf 0}})_t=\left(0,0,\epsilon\right)^{\top}$), so no input sequence could make the filter zero. Following Proposition \ref{prop:1}, we have proved that $U(\underline{{\bf z}})$ is injective for $z\in [0,1]$.
\end{example}

\begin{example}\label{ex:period}
\normalfont
We finish by illustrating the content of Proposition \ref{prop:2} when Proposition \ref{prop:1} cannot be explicitly shown. Let us define again the CPTP map $\mathcal{E}': \mathcal{B}({\mathcal H}) \longrightarrow \mathcal{B}({\mathcal H}) $
\begin{equation} 
\mathcal{E}'(A)=(1-\epsilon)A+\epsilon\tr(A)\sigma.
\end{equation}
where we will study the case of a single qubit, we fix $\sigma = \ket{0}\bra{0}$ and whose SAS representation was given in \eqref{eq:SASepsilon}.
Now we take a rotation around the $y$ axis with unidimensional input $z$ as a unitary channel $\mathcal{J}$, that is
\begin{equation}
\mathcal{J}(\rho,z)=\mathcal{U}(z)\rho\mathcal{U}^{\dagger}(z) =e^{-i\frac{z}{2}\sigma_y}\rho e^{i\frac{z}{2}\sigma_y}.
\end{equation}
The SAS representation of this quantum channel is the following:
\begin{equation}\label{eq:qpJ}
q_{\mathcal{J}}(z)=\textbf{0},\quad p_{\mathcal{J}}(z)=\begin{pmatrix}  \cos(z)& 0 & -\sin(z)\\
0 &  1 & 0 \\
\sin(z) &  0 & \cos(z)    
\end{pmatrix}.
\end{equation} 
The linear map~\eqref{eq:rankJunitary} is given by
\begin{equation}
\label{for maximal rank}
p_{\mathcal{E}'}\left(\mathcal{D}p_\mathcal{J}({z})(s) \textbf{x}\right) 
= - (1-\epsilon)\begin{pmatrix}  x _1\sin(z)+ x _3\cos(z)\\
0 \\
- x _1\cos(z)+ x _3 \sin (z)    
\end{pmatrix}s,
\end{equation}
for any $s\in \mathbb{R}$ and any $\mathbf{x}=(x _1, x _2, x _3)^{\top}  \in V$. This implies that the map $ p_{\mathcal{E}'}\left(\mathcal{D}p_\mathcal{J}({z})(\cdot )\textbf{x}\right): \mathbb{R}\rightarrow {\Bbb R}^3 $ has maximal rank equal to one as long as  $x _1\sin(z)+ x _3\cos(z) $ and $- x _1\cos(z)+ x _3 \sin (z) $ are not simultaneously equal to zero, which, as we now see is only possible when $\mathbf{x}= {\bf 0} $. Indeed, if $x _1\sin(z)+ x _3\cos(z) =0$, there exists $\alpha \in \mathbb{R} $ such that $x _1= \alpha\cos (z) $ and $x _2= -\alpha\sin(z) $. If now $- x _1\cos(z)+ x _3 \sin (z) =0$, then this implies, together with the previous relation, that $\alpha(\sin ^2(z)+\cos ^2(z))=0 $ and hence that $\alpha=0 $, which can only happen if $\mathbf{x}= {\bf 0}$. 

 This result, together with Proposition \ref{prop:2} can be used to show that the filter corresponding to this contracted-encoding channel is locally injective for inputs in the neighborhood of many input sequences. We show this by first spelling out this filter. Denote by $R _\theta ^Y: \mathbb{R}^3 \rightarrow {\Bbb R}^3 $ the rotation in three dimensions by an angle $\theta $ around the OY-axis. It is easy to see that in our case, the expression \eqref{eq:filter_x} for the SAS filter takes the form
\begin{equation*}
U(\underline{{\bf z}})_t= {\epsilon}\sum_{j=0}^{\infty} (1- \epsilon)^jR _{\sum_{k=0}^{j-1}z_{t-k}} ^Y \left(0,0,1\right)^{\top}, \ t \in \Bbb Z.
\end{equation*}
Notice that this filter is not globally injective as adding $2\pi$ to each entry $z_t$ of the input sequence produces a different input with the same image under $U$, and additionally, showing $U(\underline{{\bf z}})_0\neq \textbf{0}$ for a given $\epsilon$ seems an unfeasible task, discarding the application of Proposition \ref{prop:1} even for cases where the input does not lead to periodicity, like $D_n \subset [0,2\pi)$. However, as we now prove, $U$ is locally injective for inputs in a neighborhood of certain input sequences. Let us first tackle the case of constant input sequences, and then we will consider nonconstant sequences. Let $z  \in \mathbb{R} $ be arbitrary, let $\underline{{\bf z}}=(\dots,z,z)$ be the corresponding constant sequence, and $\textbf{x}_0=U(\underline{{\bf z}}) _0 $. Then, the linear map \eqref{eq:rankJunitary} becomes
\begin{equation}\label{eq:rank_last_ex}
\begin{split}
&p_{\mathcal{E}'}\left(\mathcal{D}p_\mathcal{J}({z})(s) \textbf{x}_0\right) 
 \\
&=-\frac{\epsilon(1-\epsilon)}{1+(1-\epsilon)^2-2(1-\epsilon)\cos(z)}\begin{pmatrix}  \cos(z)\\
0 \\
\sin(z)    
\end{pmatrix}s,
\end{split}
\end{equation}
for any $s\in \mathbb{R}$. This vector is different from zero whenever $0< \epsilon<1$. Hence, Corollary \ref{cor:1} implies the existence of a neighborhood $A_{\underline{{\bf z}}} $ of the sequence $\underline{{\bf z}} $ such that the restriction $U|_{A_{\underline{{\bf z}}} } $ is injective.

As we just showed, \eqref{eq:rank_last_ex} has maximal rank for any $z  \in \mathbb{R} $ whenever $0< \epsilon<1$, so this is trivially true for any $z_0\in\left(\mathbf{x} ^\ast\right)^{-1}(\mathbf{x} _0)$, where we take $\mathbf{x} _0 =U(\underline{{\bf z}}_0) _0$ with $\underline{{\bf z}}_0=(\dots,z_0,z_0)$ for simplicity. Now consider, for instance, the input space $D_n =[0,2\pi]$ and $z_0=0$. Due to the periodicity of the filter $U$ under each entry $z_t\in D_n$, we find that $\left(\mathbf{x} ^\ast\right)^{-1}(\mathbf{x} _0) =\{0,2\pi\}$. Therefore, Proposition \ref{prop:2} guarantees that for each input sequence $\underline{{\bf z}} _0$ constructed using the two elements $\{0,2\pi \}$ as entries, there exists a product set $A_{\underline{{\bf z}} _0}$ such that the restrictions $U|_{A_{\underline{{\bf z}} _0}} $ is injective. 

We conclude by illustrating the importance of injectivity with a numerical example that shows its effect on this model when solving tasks. We consider the short-term reconstruction task that consists of taking as target $\hat{y}_t := z_{t-1}$, that is, we ask the reservoir to recall the previously one time step-lagged fed inputs. These inputs are drawn from a random uniform distribution in the interval $z_t\in[0,1]$. The output layer $y_t\in \mathbb{R}$ of the quantum reservoir is constructed as 
\begin{equation}
    y_t = {\bf x}_t^\top{\bf w}+b,
\end{equation}
where ${\bf w}\in \mathbb{R}^{3}$ is the output weight vector, and we have added a constant offset $b\in\mathbb{R}$. Given a training set of length $N$, the ridge regularized least square approximation to some target vector $\hat{\bf y}\in\mathbb{R}^{N}$ is given by the expression
\begin{equation}
  \text{arg }\min_{\bf w}\left( \|\textbf{X}{\bf w}-\hat{\bf y}\|^2_2+\lambda\|{\bf w}\|^2_2\right),
\end{equation}
where  $\textbf{X}\in \mathbb{R}^{N\times 3}$ is the matrix of observables at all training time steps, $\lambda>0$ is the regularization hyperparameter, and $\|\cdot\|_2$ is the Euclidean norm. The explicit solution of this optimization problem is given by the expressions:
\begin{equation}
\begin{split}
    &\bar{\bf w}= ({\bf X}^\top{\bf A}_{N}{\bf X}+\lambda {\bf I})^{-1}{\bf X}^\top {\bf A}_{N}\hat{\bf y},\\
    & \bar{b} = \frac{1}{N}{\bf 1}^\top_N(\hat{\bf y}-\textbf{X}\bar{\bf w}),
    \end{split}
\end{equation}
where ${\bf I}$ is the identity, ${\bf 1}_N:= (1,\dots,1)^\top\in \mathbb{R}^N$, and ${\bf A}_{N}:={\bf I}-{\bf 1}_N{\bf 1}^\top_N/N$. To quantify the performance of the reservoir, we use the memory capacity $C:=\text{cov}^2(\hat{\bf y},{\bf y})/(\sigma^2(\hat{\bf y})\sigma^2({\bf y}))$, where $\text{cov}(\cdot)$ is the covariance and $\sigma^2(\cdot)$ is the variance. For this task, we set 100 inputs to washout the initial conditions, $N=1000$ for training, $N_{\text{test}}=1000$ for test and $\lambda=1\times 10^{-10}$ for regularization. We average $C$ over 100 random realizations of the input sequence.

To test the effect of injectivity on this task, we will force the reservoir to encode inputs out of the interval $[0,2\pi)$. Let us parametrize the input encoding channel with an input strength parameter $g$ as
\begin{equation}
\mathcal{J}(\rho,z)=\mathcal{U}(z)\rho\mathcal{U}^{\dagger}(z) =e^{-ig\frac{z}{2}\sigma_y}\rho e^{ig\frac{z}{2}\sigma_y}.
\end{equation}

\begin{figure}[h]
\captionsetup[subfigure]{}
\begin{center}
\includegraphics[scale=0.15]{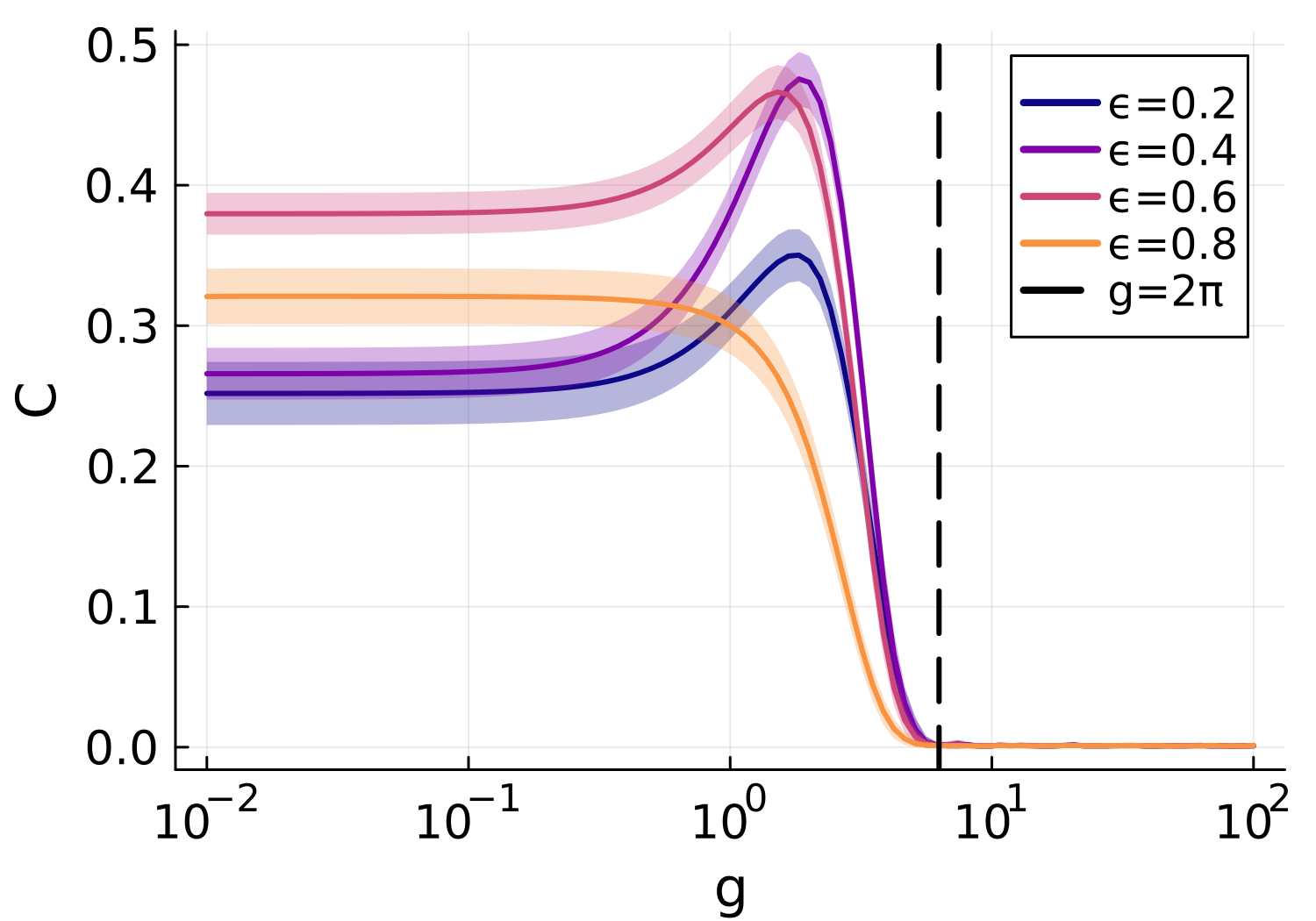}
\caption{Memory capacity for the short-term memory task $\hat{y}_t = z_{t-1}$. Solid lines denote the mean values, and shaded regions represent the standard deviations for 100 realizations of the uniformly distributed random input sequence.}\label{Fig:2}
\end{center}
\end{figure}

Figure \ref{Fig:2} displays the values of $C$ against $g$ for several values of the hyperparameter $\epsilon$. For $g\ll 1$, we observe a uniform performance, but as we increase $g$, the behavior changes. For $g\sim 1$, the performance reaches a maximum (except for $\epsilon=0.8$). After this maximum, the performance drops until we reach  $g=2\pi$, where the reservoir ceases to be functional. This critical point is exactly the value that maps the input out of the interval $[0,2\pi)$, breaking the global injectivity of the model. 

With this numerical experiment, we have shown that global injectivity is a fundamental trait to consider when solving temporal tasks, especially for periodic systems.
\end{example}

\section{Discussion}\label{sec:discussion}
The design of optimal QRC systems is a complex task that requires not only the optimization of hyperparameters but also the choice of a pertinent model right from the outset. Valuable QRC models are obtained when a quantum channel and an input encoding are selected that present adequate information processing properties. In this paper, we have built on our previous results in Ref.~\cite{martinez2023quantum} and further explored conditions that guarantee the presence of valuable input response features in QRC systems. In particular, we have provided conditions under which a QRC filter is injective, that is when it can distinguish its outputs between different input sequences. Injectivity is a fundamental property of any machine learning technique, even though not much attention has been paid to it so far in the context of reservoir computing, where it becomes a key trait in the learning of deterministic dynamical systems (see Remark \ref{remark:8}). Proposition \ref{prop:1} provides an explicit sufficient rank condition for SAS systems to have injective filters. We recall that SAS systems are the cornerstone of QRC models in finite dimensions. The condition in Proposition \ref{prop:1} needs to be tested at any possible input and any possible reachable state by the system, which may be cumbersome in some cases. This difficulty can be eased by formulating a local version of this statement in which the condition has to be verified at only one output value $\mathbf{x} _0$ {that would be obtained by feeding a constant signal consisting of ${\bf z} _0 $-inputs}. This yields an easy-to-verify criterion at the price of guaranteeing injectivity only in a neighborhood of the input sequences that lead to $\mathbf{x} _0$.

In the context of the design constraints that we just explained, we have analyzed an important family of quantum reservoir models constructed as the composition of an input-encoding CPTP map followed by a strictly contractive channel that guarantees the ESP and the FMP. We refer to the elements of this family as contracted-encoding quantum channels. The goal of our analysis is to understand how the elements of this family can be designed so that they yield optimal input dependence properties and, in particular, when the corresponding filters are injective. Theorem \ref{th:2} spells out situations that yield constant filters, and that should be avoided in contracted-encoding quantum architectures. We completed the discussion by providing an extension of Proposition \ref{prop:1} and Corollary \ref{cor:1} to this relevant family of models. The main upshot of this part of the paper is that the design of valuable QRC systems capable of processing arbitrarily long input sequences and differentiating between them requires careful analysis. In particular, both the choice of the quantum channel and the input encoding are influential in the quality of the input response.

To conclude, we are convinced that further numerical experiments could be implemented to spell out the importance of injectivity. Example \ref{ex:period} already shows how breaking injectivity can drop the performance of a reservoir in a very simple task. A more valuable numerical experiment that could be studied is the learning of dynamical systems. However, single-qubit models would not be enough, and large quantum systems would be required, with the subsequent complications of evaluating the rank condition. An approach that can be explored is numerically obtaining the Pauli matrix representation of the quantum channel following Ref.~\cite{hantzko2024pauli} and then using automatic differentiation to compute the Fréchet derivatives of $p({\bf z})$ and $q({\bf z})$. To study the rank condition of Proposition \ref{prop:1} we would need to sample from the set of reachable states $V_R\subset V$, but we can alternatively sample in the larger space $V$ if it is not possible to obtain $V_R$. A simpler situation is given in Proposition \ref{prop:2} and Corollary \ref{cor:1}, because we can explicitly compute ${\bf x}_0$  by feeding a constant input sequence. Further simplifications of the rank condition can be obtained by understanding the physics behind the quantum channel, as demonstrated for the contracted-encoding quantum channels.

\section*{Acknowledgments}
We thank Shumpei Kobayashi for inspiring and useful discussions. Interesting input from two referees is also acknowledged. R.M.P. and J.P.O. acknowledge partial financial support from the School of Physical and Mathematical Sciences of the Nanyang Technological University through the SPMS Collaborative Research Award 2023 entitled ``Quantum Reservoir Systems for Machine Learning." 
RMP acknowledges the QCDI project
funded by the Spanish Government. J.P.O. acknowledges the hospitality of the Donostia International Physics Center and R.M.P. of the Division of Mathematical Sciences of the Nanyang Technological University during the academic visits in which some of this work was developed.

\section*{Data Availability}
The data that support the findings of this article
are openly available \cite{github}. The Mathematica notebook for some of the examples, generating Fig.~\ref{Fig:1} and the Julia code to generate Fig.~\ref{Fig:2} is available at Ref.~\cite{github}.

\end{document}